\documentclass[11pt]{article}
\usepackage{graphicx} 
\usepackage{amsmath, amsthm, amssymb, bbm, mathrsfs}
\usepackage{braket, quantikz}

\usepackage[a4paper, total={6in, 9.3in}]{geometry}

\usepackage[maxbibnames=99, style=alphabetic]{biblatex}
\usepackage[colorlinks]{hyperref}
\hypersetup{
  linkcolor=[rgb]{0.3,0.3,0.6},
  citecolor=[rgb]{0.2, 0.6, 0.2},
  urlcolor=[rgb]{0.6, 0.2, 0.2}
}
\usepackage[capitalize]{cleveref}

  \newcommand{\R}{\mathbb{R}} 
  \newcommand{\Z}{\mathbb{Z}} 
  \newcommand{\C}{\mathbb{C}} 
  \newcommand{\N}{\mathbb{N}} 
  \newcommand{\F}{\mathbb{F}} 
  \newcommand{\bset}[1]{\{0,1\}^{#1}} 
  \newcommand{\U}{\mathcal{U}} 
  \newcommand{\Rcal}{\mathcal{R}} 
  \newcommand{\M}{\mathrm{M}_{\C}} 
  
  \newcommand{\Gcal}{\mathcal{G}} 
  \newcommand{\Pcal}{\mathcal{P}} 
  \newcommand{\Lcal}{\mathcal{L}} 
  \newcommand{\Aut}{\mathrm{Aut}}
  
  \newcommand{\Pol}{\mathrm{Pol}}

  \DeclareMathOperator{\tr}{tr}
  \newcommand{\thr}{\mathrm{Th}} 
  \newcommand{\Eq}{\mathrm{Eq}}
  \newcommand{\NOT}{\mathrm{NOT}} 
  \newcommand{\one}{\mathbf{1}} 
  \newcommand{\id}{\mathrm{Id}} 

  \newcommand{\Ccal}{\mathcal{C}} 
  
  \newcommand{\Scal}{\mathcal{S}}
  \newcommand{\Acal}{\mathcal{A}} 
  \newcommand{\Dcal}{\mathcal{D}} 
  \newcommand{\Hcal}{\mathcal{H}} 
  \newcommand{\vspan}{\mathrm{span}}

  \theoremstyle{plain}
  \newtheorem{theorem}{Theorem}
  \newtheorem{lemma}{Lemma}
  
  \newtheorem{prop}{Proposition}
  \newtheorem{corollary}{Corollary}
  
  \theoremstyle{definition}
  \newtheorem{definition}{Definition}
  \newtheorem{problem}{Problem}
  \newtheorem{example}{Example}

  \theoremstyle{remark}
  \newtheorem*{remark}{Remark}

\title{Symmetric quantum computation}
\author{Davi Castro-Silva\thanks{University of Cambridge}
\and Tom Gur\footnotemark[1]
\and Sergii Strelchuk\thanks{University of Oxford}}

\begin{document}

\maketitle

\begin{abstract}
    We introduce a systematic study of \emph{symmetric quantum circuits}, a new restricted model of quantum computation that preserves the symmetries of the problems it solves.
    This model is well-adapted for studying the role of symmetry in quantum speedups, extending a central notion of symmetric computation studied in the classical setting.

    Our results establish that symmetric quantum circuits are fundamentally more powerful than their classical counterparts.
    First, we give efficient symmetric circuits for key quantum techniques such as \emph{amplitude amplification}, \emph{phase estimation} and \emph{linear combination of unitaries}.
    In addition, we show how the task of \emph{symmetric state preparation} can be performed efficiently in several natural cases.
    Finally, we demonstrate an \textit{exponential separation} in the symmetric setting for the problem
    XOR-SAT, which requires exponential-size symmetric classical circuits but can be solved by polynomial-size symmetric quantum circuits.
\end{abstract}

\section{Introduction}

\begin{flushright}
\rightskip=2cm\textit{``Symmetry, as wide or as narrow as you may define its \\ meaning, is one idea by which man through the ages has tried \\ to comprehend and create order, beauty, and perfection.''} \\
\vspace{.2em}
---Hermann Weyl
\end{flushright}
\vspace{1em}

Many important computational problems possess rich groups of symmetries under which they are invariant.
In the context of machine learning, the classification of images from a data set into either ``cat pictures'' or ``dog pictures'' should be indifferent to the translation or rotation of pixels in those images.
Problems that are defined over abstract data structures -- such as graph problems in combinatorial optimisation -- will likewise be invariant with respect to the intrinsic symmetries of those data structures.
Physics provides another rich stream of problems for which symmetry considerations are essential;
indeed, the invariance of physical processes under certain special transformations forms the basis of several physical theories.

Time and again it has been proven crucial to incorporate information about the symmetries of the problems when designing efficient algorithms to solve them.
In the words of Bronstein et al \cite{BronsteinBCV2021}:
``Exploiting the known symmetries of a large system is a powerful and classical remedy against the curse of dimensionality'';
this remedy has been applied with great success in all three areas mentioned above, as well as several others.

In the area of machine learning, such considerations have sparked the nascence of the field of \emph{geometric deep learning}, where symmetries are incorporated as geometric priors into the learning architectures so as to improve trainability and generalisation performance \cite{BronsteinBCV2021}.
In the area of combinatorial optimisation, the symmetries of semidefinite programs are exploited in order to block-diagonalise their formulations so as to greatly reduce the dimension of the resulting program \cite{BachocGSV2012}.
In the fields of sublinear algorithms and property testing, the richness of the automorphism group of the problem often determines its complexity \cite{KaufmanS2008}.
Throughout physics, symmetry-preserving ans\"{a}tze are used as starting points when searching for solutions to a complex problem \cite{skolik2023equivariant}.

The main intuition underlying the applications above is simple:
by taking the symmetries of a computational problem into account when searching for its solution, one can vastly reduce the effective dimension of the search space without impacting the quality of the solution.
Consequently, the presence of extensive symmetries is a strong indicator that a problem may admit an efficient algorithm.
Identifying such general principles to predict algorithmic efficiency is a central goal of complexity theory and algorithm design.

This paper is concerned with \emph{quantum computing} and its relative power when compared to classical computation.
A central challenge in this field is to identify which problems allow for a significant quantum speedup, and what general principles govern such an advantage.
It is therefore natural to ask how a problem's symmetries impact -- either positively or negatively -- the potential for quantum advantage.
This question is the primary motivation for our work.

Our first contribution is conceptual: we introduce a framework of \emph{symmetric quantum circuits}, which formally characterises what we mean by performing a symmetric computation on a quantum computer.
Loosely speaking, in this framework we require the symmetries of the underlying problem to be respected by the quantum circuits that solve it.
This model is well-adapted for determining the role of symmetries in quantum speedups, extending a central notion of symmetric computation introduced in the context of logic and descriptive complexity \cite{AndersonD2017, Dawar2020}.
Note that, as an abundance of symmetries is helpful in designing both classical \emph{and} quantum algorithms, it is a subtle and intriguing question whether it contributes to quantum advantage.

We then study the power of symmetric quantum computation, proving several results that showcase the (somewhat surprising) strength of this computational paradigm.
To explain those results, we first give an informal description of our model.

\subsection{Symmetric quantum circuits}
\label{subsec:QuCircuits}

Suppose we are dealing with a symmetry group $\Gamma$ which acts on a finite set $X$.
For instance, $X$ could be a set of labels for some collection of objects and~$\Gamma$ could be the permutation group~$\Scal_X$;
being symmetric under~$\Gamma$ means that the problem is invariant under any relabelling of the objects.
Another important example concerns \emph{graph properties} -- such as connectedness, triangle-freeness and 3-colourability.
These are properties that depend only on the abstract structure of the graph, and not on its representation;
in such cases, we can take $X$ to be all pairs of vertices representing the potential edges in the graph, and $\Gamma$ to be the permutation of pairs induced by vertex permutations (i.e., graph isomorphisms).

We are interested in computing functions over~$X$ which are symmetric with respect to~$\Gamma$.
However, we do not wish to introduce any extraneous breaking of symmetries during the computation, but instead require that the whole computation respect the symmetries of the problem.
In other words, we do not rely on \emph{asymmetry} as a computational resource.

To better illustrate our notion of symmetry in circuits, we exemplify it using a simple state discrimination problem:
to distinguish whether two single-qubit states $\ket{\phi}$, $\ket{\psi} \in \C^2$ are equal or have inner product at most $\delta<1$ (in modulus), under the promise that one of these two cases holds.
This task involves only two qubits, and it is naturally symmetric under exchanging the roles of $\ket{\phi}$ and $\ket{\psi}$;
we can then take $X = \{1, 2\}$ to label the relevant qubits and take the symmetry group $\Gamma$ to be $\Scal_2$
(the permutation group on two elements).

It is well known that the problem above can be solved with (one-sided) bounded error by the SWAP test proposed in \cite{BuhrmanCWW2001}, as shown below:
\begin{equation*}
    \centering
 \includegraphics[width=7cm]{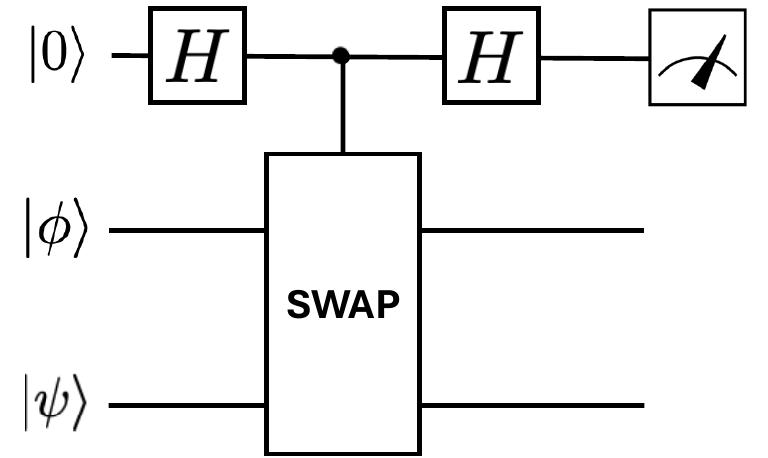}
\end{equation*}

In this circuit, measuring the first qubit produces outcome~$1$ with probability
$\frac{1}{2} - \frac{1}{2} |\langle\phi | \psi\rangle|^2$.
This probability is zero if $\ket{\phi} = \ket{\psi}$, and it is at least $\frac{1}{2} (1-\delta)^2$ if $|\langle\phi | \psi\rangle| \leq \delta$.
Thus, the test determines which case holds with one-sided error $\frac{1}{2} (1+\delta^2)$.

Importantly for us, note that the circuit above is \emph{symmetric} with respect to $\Gamma = \Scal_2$:
since the SWAP gate is symmetric, exchanging the roles of $\ket{\phi}$ and $\ket{\psi}$ does not alter the circuit.\footnote{This circuit is not symmetric with respect to all qubits, but the symmetry of the problem is only relevant for the ``active'' qubits labelled by $\ket{\phi}$ and $\ket{\psi}$.}
This would serve as an illustrative example of a quantum circuit that respects the symmetries of the problem it solves, were it not for one issue:
the circuit uses a (controlled) SWAP gate, which is not typically considered to be an elementary gate.
Instead, a SWAP gate can be implemented using three CNOT gates, which \emph{are} normally seen as elementary.

Suppose that our gate set includes only single-qubit gates, CNOTs and Toffoli gates.
Using the usual implementation of a SWAP gate as three CNOTs, we arrive at the following equivalent circuit:
\begin{equation*}
    \includegraphics[width=8cm]{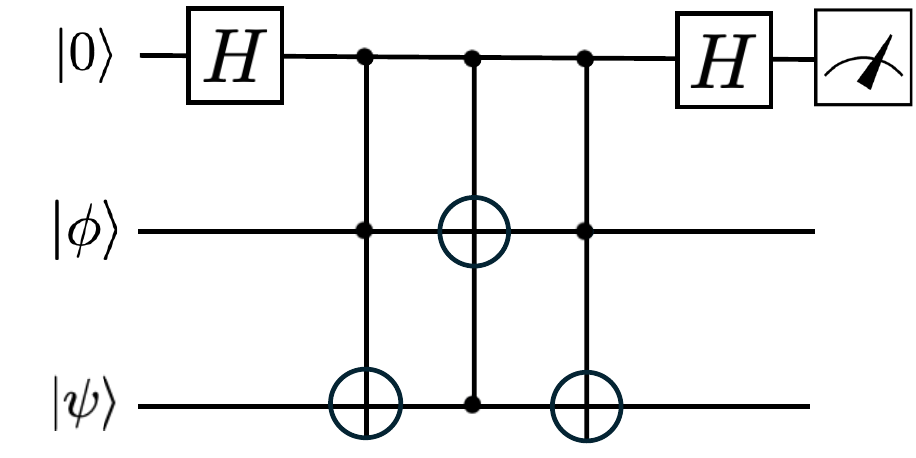}
\end{equation*}

Now we have a problem:
even though this last circuit is precisely equivalent to the one we started with, it no longer preserves the same symmetries!
Indeed, the wire carrying $\ket{\phi}$ can be distinguished from the one carrying $\ket{\psi}$ because it contains only one ``head'' of a Toffoli gate.

One might wonder whether it is at all possible to implement the SWAP test symmetrically using our specified gate set.
Indeed, this appears to be impossible if we only have access to three wires as above.
However, if we allow for an extra ancilla qubit, then one can implement the SWAP test circuit as shown below:
\begin{equation*}
    \centering
    \includegraphics[width=12cm]{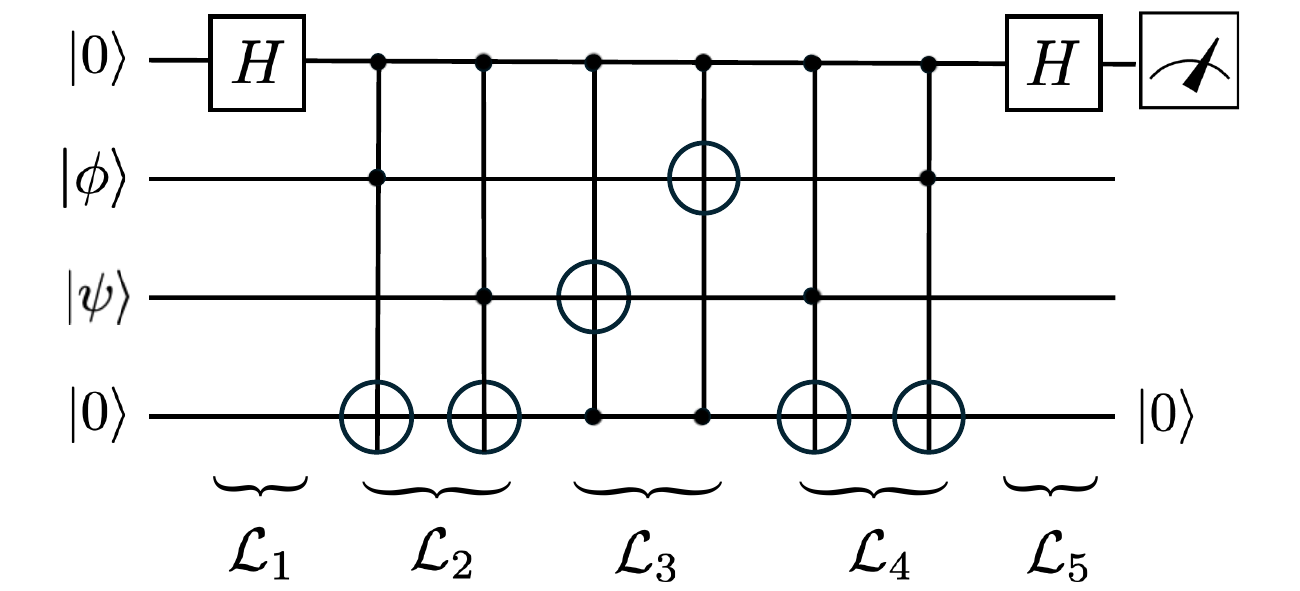}
\end{equation*}

This last circuit implements exactly the same operation as the original SWAP test, but now using an extra ancilla qubit
(which is returned to its initial state $\ket{0}$ at the end of the computation).
Its gates are grouped into five layers, where the gates inside any given layer commute and thus have no intrinsic order in their implementation.
Crucially, this circuit is now symmetric:
changing the roles of $\ket{\phi}$ and $\ket{\psi}$ only permutes the Toffoli gates inside each layer;
since there is no order of implementation of gates within a layer, this represents an automorphism of the circuit.
(For the precise definition of circuit automorphism, see Definition~\ref{def:automorphism}.)
This example motivates the following notion of symmetry in circuits.

Let~$X$ be a finite set and $\Gamma$ be a symmetry group acting on~$X$.
A \emph{$\Gamma$-symmetric quantum circuit} will have qubits labelled by the elements of $X$ (the ``active qubits'') and additional ``workspace qubits'' (ancillas) to help with the computation.
As usual, the workspace qubits are initialised to $\ket{0}$ while the input to the circuit is encoded in the active qubits.
We require the circuit to be \emph{invariant} under the induced action of $\Gamma$:
any symmetry acting on the active qubits must be extendable to an automorphism of the circuit itself. 
More precisely, for any $\sigma\in \Gamma$ there must exist a permutation $\pi$ of the workspace qubits such that permuting the active (resp., workspace) qubits by applying $\sigma$ (resp., $\pi$) we get an isomorphic circuit, as in the example above.

It is clear that symmetric circuits can only compute symmetric maps, meaning those whose value is unchanged by permutations of the input corresponding to a symmetry in $\Gamma$.
The reverse direction is more subtle, and it becomes especially hard when the complexity of implementing the circuits is taken into account (which is indeed our goal).

While the notion of symmetric circuits can be defined in a way that is agnostic to the gate set, we will study those equipped with a gate set formed by unbounded-fan-in \emph{threshold gates} together with arbitrary single-qubit gates.
We argue that this is a natural family of gates to consider for the following reasons:
\begin{itemize}
    \item These gates are mathematically simple, and they generate a robust gate set for symmetric circuits.
    \item They form a \emph{minimal} quantum extension of the classical notion of symmetric circuits \cite{Dawar2020, Dawar2024}.
    Indeed, if we allow only single-\emph{bit} reversible gates instead of arbitrary single-qubit gates, we arrive precisely at the standard notion of symmetric (threshold) circuits that has been studied in the classical setting.
    \item Threshold gates are very natural from the perspective of classical circuits, forming the basis for the complexity classes TC$^k$.
    In the quantum setting, they were shown to be efficiently implemented by bounded-depth quantum circuits equipped with two-qubit gates and fan-out gates \cite{HoyerS2005}.
    Moreover, threshold gates were shown to have the same power as fan-out in the setting of bounded-depth quantum circuits \cite{GrierM2024l}.
    They were also shown to be efficiently implemented in the hybrid quantum class LAQCC \cite{Buhrman2024statepreparation}.
    \item The class of symmetric circuits generated by them is rather powerful, as we expand upon below.
\end{itemize}

\subsection{On the power of symmetric quantum circuits}


Equipped with the foregoing notion of symmetric quantum circuits, we proceed to study the power of symmetric quantum computation, proving several results that showcase the strength of this computational paradigm.

Our first result provides efficient symmetric circuits for central quantum subroutines.
In what follows, we denote by~$s(\Ccal)$ the size (i.e., number of gates) of a circuit~$\Ccal$.
Whenever a symmetry group~$\Gamma$ is left undefined, the result holds uniformly over all groups that~$\Gamma$ could represent.

\begin{theorem}[Symmetric quantum subroutines]
\label{thm:subroutines}
The following procedures can be efficiently implemented by symmetric quantum circuits:
\begin{enumerate}
    \item \textbf{Amplitude amplification.}
    Let~$\Acal$ and~$\Dcal$ be $\Gamma$-symmetric circuits, where~$\Acal$ prepares a state as follows
    $$\Acal\ket{0}^n = \sqrt{p} \ket{\psi_1} + \sqrt{1-p} \ket{\psi_0},$$
    while~$\Dcal$ flags~$\ket{\psi_1}$ in the sense that~$\Dcal \ket{\psi_i}\ket{0} = \ket{\psi_i}\ket{i}$ for $i \in\{0, 1\}$.
    Then we can construct a $\Gamma$-symmetric circuit that maps $\ket{0}^n$ to $\ket{\psi_1}$ using $O\big((s(\Acal)+s(\Dcal))/\sqrt{p}\big)$ gates.

    \item \textbf{Phase estimation.}
    Let~$\Acal$ be a $\Gamma$-symmetric circuit and $0<\varepsilon<1$.
    We can implement a $\Gamma$-symmetric circuit of size $O\big(s(\Acal)/\varepsilon + 1/\varepsilon^2\big)$ which estimates the eigenvalue of any given eigenstate~$\ket{\psi}$ of~$\Acal$ up to accuracy~$\varepsilon$.
   
    \item \textbf{Linear combination of unitaries.}
    Let $U_0, \dots, U_{k-1}$ be $n$-qubit unitaries, let $\alpha_0, \dots$, $\alpha_{k-1} \in \C$ and denote
    $$L := \bigg\|\sum_{i=0}^{k-1} \alpha_i U_i \ket{0}^n\bigg\|.$$
    Given $\Gamma$-symmetric circuits $\Ccal_0, \dots, \Ccal_{k-1}$ such that each $\Ccal_i$ implements $\ket{0}^n \mapsto U_i \ket{0}^n$,
    we can construct a $\Gamma$-symmetric circuit that implements
    $$\ket{0}^n \mapsto \frac{1}{L} \bigg(\sum_{i=0}^{k-1} \alpha_i U_i \bigg) \ket{0}^n$$
    using $O\big(\|\alpha\|_1 \big(k\log k + \sum_{i=0}^{k-1} s(\Ccal_i)\big)/L\big)$ gates.
    \end{enumerate}
\end{theorem}

An important example of how symmetry can be useful in quantum information theory and quantum computing is given by the \emph{symmetric subspace} $\vee^n \C^2$.
This is the vector space formed by all $n$-qubit states which are invariant under permutations of the qubits, and has been proven useful in state estimation, optimal cloning and the quantum de Finetti theorem;
see \cite{Harrow2013} for an account of these applications.
We show that every state in the symmetric subspace can be efficiently prepared by a permutation-symmetric circuit:

\begin{theorem}[Symmetric state preparation]
\label{thm:sym_state_intro}
    Given the classical description of an $n$-qubit symmetric state $\ket{\varphi} \in \vee^n \C^2$, we can construct an $\Scal_n$-symmetric quantum circuit which prepares $\ket{\varphi}$ using $O(n^{2.75})$ gates.
\end{theorem}

Note that elements of the symmetric subspace are the only states which can be prepared by an $\Scal_n$-symmetric circuit (when starting from~$\ket{0}^n$), even in information-theoretic terms.
The last theorem then provides an equivalence between what is \emph{possible} to prepare and what is \emph{efficient} to prepare in the setting of $\Scal_n$-symmetric circuits, a result which has no clear analogue in the non-symmetric setting.

In addition to preparing states belonging to the symmetric subspace $\vee^n \C^2$, similar techniques to those used in Theorem~\ref{thm:sym_state_intro} can also be applied for efficiently implementing the full action of a given permutation-symmetric unitary on $\vee^n \C^2$.
This is shown by our next theorem:

\begin{theorem}[Symmetric subspace unitaries] \label{thm:sym_subspace_intro}
    Given a permutation-symmetric $n$-qubit unitary $U$, we can implement the restricted action of~$U$ on the symmetric subspace by an $\Scal_n$-symmetric circuit with $O(n^{3.75})$ gates.
\end{theorem}

Another important setting where symmetry considerations prove crucial within quantum computing is the setting of \emph{quantum machine learning}.
In that context, Schatzki et al \cite{SLNSC2024} recently proposed a class of \emph{group-equivariant quantum neural networks} that encode the symmetries of the considered problem into their learning architectures.
Here we show how their proposed class of QNNs naturally falls within our framework of symmetric quantum circuits.

\begin{theorem}[Equivariant QNNs] \label{thm:QNN_intro}
    Permutation-equivariant quantum neural networks (as defined in~\cite{SLNSC2024}) can be implemented by $\Scal_n$-symmetric circuits with only a linear increase in the number of gates.
\end{theorem}

We note that the class of equivariant QNNs proposed by Schatzki et al can be generalised to groups beyond the full permutation group~$\Scal_n$, and their translation into symmetric circuits proceeds similarly to the proof of Theorem~\ref{thm:QNN_intro}.
Our notion of symmetric circuits is then mirrored in the quantum machine learning literature, and one can regard $\Gamma$-equivariant QNNs as examples of $\Gamma$-symmetric circuits;
see Section~\ref{sec:examples} for a discussion.

Finally, we provide an example of a decision problem for which one obtains unconditional quantum advantage in the symmetric setting.
The problem XOR-SAT is a Boolean satisfiability problem where each clause is given by the exclusive OR of some subset of the variables (or their negations).
This problem is clearly symmetric with respect to permuting clauses and relabelling variables, that is, it is $\Scal_m \times \Scal_n$-symmetric
(where~$n$ is the number of variables and~$m$ is the number of clauses).

Despite the fact that XOR-SAT is solvable in polynomial time (e.g., via Gaussian elimination), it was proven by Atserias, Bulatov and Dawar \cite{AtseriasBD2009} that it cannot be solved by polynomial-size symmetric classical circuits.
This lower bound was later made more explicit and strengthened by Atserias and Dawar \cite{AtseriasD19}.
Here we show that symmetric quantum circuits can solve this problem in polynomial time.

\begin{theorem}[Symmetric quantum advantage] \label{thm:advantage_intro}
    The problem XOR-SAT on~$n$ variables and~$m$ clauses can be solved (with certainty) by $\Scal_m \times \Scal_n$-symmetric quantum circuits of size $\Tilde{O}(m^2 n^2)$.
    By contrast, any $\Scal_m \times \Scal_n$-symmetric threshold circuit that solves XOR-SAT with $m \geq n$ must have size~$2^{\Omega(n)}$.
\end{theorem}


\subsection{Outlook}

This paper initiates the study of symmetric quantum circuits, with the aim of studying the relationship between symmetries and quantum speedups.
In order to do so, we introduce a new \emph{restricted model of computation} -- akin to the query model \cite{Ambainis2018}, where the input to a problem can only be accessed through black-box queries, or to shallow circuits \cite{BravyiGK2018}, where computation is restricted to circuits of bounded depth.
In contrast to those two other settings, however, the restriction in our model is not uniform across all problems but is instead \emph{problem-specific}, depending on the symmetries of the task at hand.

We stress that each of the restricted computational models mentioned above has its own advantages and disadvantages, and they are fundamentally incomparable to each other.
Each model highlights a different facet of quantum computation, and the specific insights gleaned from any one model can be lost when translating the same problem to a different setting.
While limitations on the power of \emph{unrestricted} quantum computation are still beyond our grasp, it is only through a combination of insights obtained from each restricted model that we can get a more complete picture.

A significant advantage of classical symmetric circuits is that they provide a natural model for proving exponential lower bounds for the complexity of certain computational problems.
This raises two intriguing possibilities in the quantum setting:
that these lower bound techniques may be extended so that we can prove strong \emph{quantum} lower bounds for symmetric computation;
and that we can unconditionally prove an \emph{exponential separation} between the power of classical and quantum symmetric computation.
In this paper we show that this second possibility is indeed true, while the first possibility remains open and is left for future work.
In Section~\ref{sec:Open} we provide several other open questions and possible lines of study which will further our understanding of the role of symmetries in quantum computation.

\subsection{Related work}

The literature regarding the influence of symmetries in computational complexity and algorithm design is vast.
There are consequently many papers which are related to the present work, either directly or tangentially, and here we mention the ones we believe are most relevant to our line of investigation.

The closest line of work to ours is certainly that of (classical) symmetric computation and symmetric threshold circuits \cite{AndersonD2017, Dawar2020, Dawar2024}.
This is far from coincidental:
our paper builds upon that work to extend it to the quantum setting, and we rely on their insights both for inspiration and as a source of motivation for what we do.
Indeed, the great success of their computational model in solving nontrivial problems combined with the existence of (unconditional) exponential circuit lower bounds were what drove us to start this work.
We were also motivated by the close connection between classical symmetric computation and logic \cite{AndersonD2017}, which was essential in obtaining the aforementioned lower bounds;
we hope to bring such connections and techniques to the quantum setting as well.
The relevant concepts and results in this area will be briefly surveyed in Section~\ref{sec:classical}.

Instances of symmetric quantum circuits (in the sense considered in this paper) were implicitly explored in the context of quantum machine learning, where they lead to efficient equivariant quantum neural networks~\cite{LSFVCC2022} and other variational quantum machine learning primitives \cite{MMGMAWE2023}.
Properties of symmetric circuits were also used to design efficient classical simulation algorithms for a broad range of quantum systems~\cite{AnschuetzBBL2023}.

Symmetry considerations are often useful in optimisation tasks, as they can significantly reduce complexity, speed up computations and improve solution quality by reducing the search space or simplifying the problem structure and optimisation landscape.
For instance, in neural network optimisation, taking advantage of parameter space symmetries (where different parameter values lead to the same loss) via the so-called symmetry teleportation~\cite{zhao2022symmetry} increases the likelihood of moving to a point where gradient descent is more efficient, due to a flatter or more favourable curvature of the loss function.
Similarly, symmetry-restricted quantum circuits used to realize equivariant unitary transformations remarkably lead to rigorous performance guarantees~\cite{zheng2023speeding} and an improved runtime for quantum variational algorithms \cite{MMGMAWE2023, ragone2022representation}.

Finally, the role of symmetries in quantum speedups has been studied in the \emph{query model} of computation \cite{AaronsonA2014, Chailloux2019, BenCGKPW2020}.
Those results have mainly been negative, showing a significant amount of symmetries \emph{precludes} super-polynomial quantum advantage in the number of queries used.
Symmetries have also been shown to limit the performance of variational algorithms~\cite{bravyi2020obstacles}, and to impose additional constraints on the evolution of quantum systems with local interactions \cite{Marvian2022, MarvianLH2022}.

To avoid confusion, we would like to elaborate on how our work differs from the latter two. Despite common nomenclature, our work and that of Marvian, Liu, and Hulse explores the role of symmetry in quantum systems from two fundamentally different perspectives, leading to seemingly contradictory but ultimately complementary conclusions. The primary distinction lies in the objectives of the two research directions: our work presents a prescriptive, computational framework for designing algorithms, while their work provides a descriptive, physical framework that identifies the inherent limitations of quantum dynamics. Our framework of symmetric quantum circuits is a model of computation established to study the prospects of quantum advantage in a complexity-theoretical setting. 
We primarily focus on discrete permutation symmetries relevant to combinatorial problems, making use of non-local computational tools like unbounded-fan-in threshold gates and the Linear Combination of Unitaries subroutine.
A key feature is the use of ancilla qubits, which offer the flexibility needed to ensure the overall circuit respects the problem's symmetry, even if internal operations are asymmetric. In contrast, their work focuses on continuous symmetries, such as $SU(d)$, which are linked to conservation laws like the conservation of total spin. The central assumption is that all physical interactions are local. This descriptive approach seeks to determine which quantum operations are physically realizable in nature, rather than what is theoretically computable. Our work proves that for certain discrete symmetries, any symmetric unitary operation can be implemented by a symmetric circuit. The work of~\cite{Marvian2022, MarvianLH2022} however, establishes a powerful no-go theorem: generic symmetric unitaries cannot be implemented using only local symmetric interactions. Our model achieves universality precisely by using computational resources which are explicitly forbidden by the locality constraint in the physical model of those authors.

\subsection{Open problems}
\label{sec:Open}

Our work opens a number of questions that can shed further light on the computational power and limitations of symmetric quantum computation.

The peculiar structure of this class of quantum computations may lead to stronger circuit lower bounds which are simpler to obtain.
This possibility is motivated by the availability of powerful methods for proving circuit lower bounds in the classical symmetric setting.
It would be very interesting to extend such methods to the quantum setting.

Another interesting question is whether the definition of symmetric circuits is robust to changes in the quantum gate set.
Theorem~\ref{thm:part_sym} shows that \emph{reversible} symmetric circuits are fairly robust to changes in the allowed gate set, and its proof extends to the quantum setting when we consider only changes in the ``classical part'' of the gate set.
It seems plausible that allowing for arbitrary $k$-qubit permutation-symmetric unitaries (for any constant~$k$) should similarly be of little consequence to symmetric quantum complexity, but proving this seems rather challenging due to the unconventional symmetry constraints.

The quantum setting also gives rise to a richer class of symmetric problems than the classical setting we build upon.
In particular, there are important discrete translation-invariant systems whose symmetries differ from the ``relational structure''-type of symmetries so prevalent in the classical setting.
Such discrete symmetries feature prominently in quantum chemistry applications and pose a number of unique challenges \cite{dovesiCRSO2005}.
How can we work with these systems from the point of view of symmetric circuits?

An important complexity-theoretic question is whether symmetric quantum circuits are universal for implementing symmetric operations:
Given some finite group~$\Gamma$, can we perform arbitrary $\Gamma$-symmetric unitaries using $\Gamma$-symmetric circuits?
It seems plausible that this is the case (for any group $\Gamma$), and we leave it as an open question for future work.

Finally, even though we know that permutation-symmetric circuits are universal for permutation-symmetric operations (as shown in Theorem~\ref{thm:qu_part_sym}), there still remains the question of their complexity:
Given a general permutation-symmetric unitary on~$n$ qubits, what is the cost of implementing it using $\Scal_n$-symmetric circuits?
Answering this question would provide the first symmetric-circuit upper bounds concerning a full class of symmetries.

\subsection{Organisation}

In Section~\ref{sec:symmcirc} we introduce the notion of symmetric \emph{classical} circuits and give an overview of the most important results in that setting.
We then show how to extend it to \emph{reversible} circuits as a stepping stone, and finally generalise the notion of symmetric circuits to the quantum setting in Section~\ref{sec:QuCircuits}, formalising the intuition presented in the introduction.

Section~\ref{sec:build_block} collects several useful building blocks for symmetric quantum algorithms, which will later be used in our main results.
In Section~\ref{sec:subroutines} we show how to implement important quantum subroutines in a symmetric way, in particular providing a proof of Theorem~\ref{thm:subroutines}.
Section~\ref{sec:state_prep} is concerned with symmetric circuits for quantum state preparation, and will contain the proof of Theorem~\ref{thm:sym_state_intro} as well as other efficient protocols for symmetric state preparation.
The exponential separation between classical and quantum symmetric computation (Theorem~\ref{thm:advantage_intro}) will be proven in Section~\ref{sec:advantage}.
We conclude in Section~\ref{sec:examples} by providing examples of symmetric circuits for implementing certain symmetric unitaries, in particular proving Theorem~\ref{thm:sym_subspace_intro} and Theorem~\ref{thm:QNN_intro}.

\section{Symmetric circuits}
\label{sec:symmcirc}

In this section, we build towards a formal definition of symmetric quantum circuits, show that it is robust, and discuss its properties.
We begin in Section~\ref{sec:classical} by recalling the classical definition of symmetric circuits, and in Section~\ref{sec:ClassicalResults} we give a brief overview of some important results in that setting.
We then proceed in two steps:
in Section~\ref{sec:reversible} we restate symmetric circuits in terms of reversible computation, and finally we generalise this model to the quantum setting and discuss its properties in Section~\ref{sec:QuCircuits}.

\subsection{Symmetric classical circuits}
\label{sec:classical}

Building towards our notion of symmetric quantum circuits, we first recall the definition of \emph{symmetric threshold circuits} that was introduced by Anderson and Dawar \cite{AndersonD2017}.
Our presentation will be somewhat informal;
we refer the reader to the original paper for full details.

A Boolean circuit is typically represented by a directed acyclic graph, where each vertex represents a gate.
The vertices without incoming edges are called \emph{input gates}, while the other vertices are labelled by Boolean operators coming from a specified \emph{gate set}.
The \emph{output gates} are those vertices having no outgoing edges.

The gate set we will use throughout this section is the one containing the $\NOT$ function and arbitrary \emph{threshold functions}
\begin{equation} \label{eq:thr}
    \thr_{\geq t}^n: \bset{n} \to \bset{}, \quad \thr_{\geq t}^n(a_1, \dots, a_n) = \one\{ a_1 + \dots + a_n \geq t\}
\end{equation}
for integers $n\geq 1$ and $t\geq 0$.
For convenience, we will also allow equality gates of the form
$$\Eq_t^n: \bset{n} \to \bset{}, \quad \Eq_t^n(a_1, \dots, a_n) = \one\{ a_1 + \dots + a_n = t\},$$
although it is easy to compute such gates using the others.
We denote this gate set by~$G_{\mathrm{thr}}$.

\begin{remark}
    While the choice of gate set might seem somewhat arbitrary, there are good reasons to focus on~$G_{\mathrm{thr}}$ for symmetric computation.
    Indeed, note that the operators in any gate set must be \emph{permutation-invariant}:
    directed graphs impose no structure on the inputs of any gate, and so each gate must compute a function that is invariant under any ordering of its inputs.
    One can show that allowing for more general permutation-invariant functions into the gate set~$G_{\mathrm{thr}}$ only modestly changes the complexity of implementing any given function using symmetric circuits;
    see \cite[Theorem~3.17]{DawarW2021} or Theorem~\ref{thm:part_sym} below.
    On the other hand, the standard Boolean basis (comprising \textrm{AND, OR} and \textrm{NOT}) is too restrictive and cannot even compute parity using polynomial-size symmetric circuits;
    see for instance \cite{AndersonD2017}.
\end{remark}

We will work with circuits over some specified finite structure~$X$ which represents the input to the problem.
In the context of classical symmetric circuits, examples of this structure~$X$ are graphs, hypergraphs, matrices, and various relational structures.
More generally, it could be any set having a well-defined symmetry group $\Gamma\leq \Scal_X$ (which will give meaning to the notion of symmetric circuits later on).
The notion of a threshold circuit over structures is defined as follows:

\begin{definition}[Threshold circuit]
    Let~$X$ be a finite set.
    A \emph{threshold circuit over~$X$} is a directed acyclic graph with a labelling, where each vertex of in-degree~0 is labelled by an element of~$X$ and each vertex of in-degree greater than~0 is labelled by an element $g\in G_{\mathrm{thr}}$ such that the arity of~$g$ matches the in-degree of the vertex.
\end{definition}

To incorporate symmetries into our circuits, we first need to specify a notion of \emph{circuit automorphism}, meaning those transformations that preserve the structure of the circuit.
Put simply, those are automorphisms of the underlying directed graph which preserve the labels of all non-input gates.
More formally:

\begin{definition}[Circuit automorphism]
    Let~$C$ be a threshold circuit over~$X$ with vertex set~$V$, edge set~$E$ and labelling $\lambda: V \to X \cup G_{\mathrm{thr}}$.
    An automorphism of~$C$ is a permutation $\pi \in \Scal_V$ such that:
    \begin{itemize}
        \item $(\pi(u), \pi(v))\in E$ whenever $(u, v)\in E$;
        \item $\lambda(\pi(v)) = \lambda(v)$ for all $v\in V$ of in-degree greater than~0.
    \end{itemize}
\end{definition}

Note that any circuit automorphism~$\pi$ must map input gates to input gates.
If we denote the set of input gates by~$V_0$, it follows that the restriction $\pi_{|V_0}$ is an element of~$\Scal_{V_0}$.
This restriction can be associated with a permutation of~$X$ via the labelling of the input gates:
there is a unique element $\sigma\in \Scal_X$ such that $\lambda(\pi(v)) = \sigma(\lambda(v))$ for all $v\in V_0$;
we then say that~$\pi$ \emph{extends}~$\sigma$.

If the permutation~$\sigma$ thus obtained is a symmetry of the structure~$X$, then we think of the circuit~$C$ as respecting that symmetry.
Indeed, this means there exists a permutation of the vertices of~$C$ which changes neither gate labels nor the graph structure, and whose action on the input gates corresponds to a symmetry of the input.
This motivates the following definition:

\begin{definition}[Symmetric threshold circuit]
    Let~$C$ be a threshold circuit over~$X$ and $\Gamma \leq \Scal_X$ be a subgroup.
    We say that~$C$ is \emph{$\Gamma$-symmetric} if the action of every $\sigma\in \Gamma$ can be extended to an automorphism of~$C$.
\end{definition}

As a simple but illustrative example, let us consider the problem of determining whether a given graph is triangle-free:
    
\begin{example}[Triangle-freeness]
\label{ex:triangle_bool}
    Suppose we wish to construct a symmetric circuit over $n$-vertex graphs that decides whether some given input graph contains a triangle
    (i.e., three vertices pairwise connected by edges).
    A natural choice for the underlying structure~$X$ consists of all $\binom{n}{2}$ potential edges of an $n$-vertex graph, so that
    $$X = \big\{ \{i, j\}:\, 1\leq i< j\leq n\big\}.$$
    (Similar choices, such as the set of all ordered pairs $[n]^2$, would translate to similar symmetric circuits.)
    The intrinsic symmetries of this structure are then permutations of pairs induced by vertex permutations, that is
    $$\Gamma = \big\{\{i, j\} \mapsto \{\tau(i), \tau(j)\}:\, \tau \in \Scal_n\big\}.$$

    A $\Gamma$-symmetric circuit for this problem can be implemented using $\binom{n}{2} + \binom{n}{3} + 1$ gates as follows.
    For each pair $\{i, j\} \in X$ we create an input gate $v_{\{i, j\}}$ which we label by $\{i, j\}$;
    for each triple $\{i,j,k\} \subseteq [n]$ we create a gate $u_{\{i,j,k\}}$ which we label by~$\thr^3_{\geq 3}$;
    and finally, we create a single output gate $out$ with label $\thr^{\binom{n}{3}}_{\geq 1}$.
    (Recall the definition of $\thr^n_{\geq t}$ in equation~\eqref{eq:thr}.)
    For each triple $\{i,j,k\} \subseteq [n]$, we then add the four directed edges $(v_{\{i,j\}}, u_{\{i,j,k\}})$, $(v_{\{j,k\}}, u_{\{i,j,k\}})$, $(v_{\{i,k\}}, u_{\{i,j,k\}})$ and $(u_{\{i,j,k\}}, out)$.
    One can check that, when the input gates are initialised to encode the adjacency relationship of an $n$-vertex graph (with arbitrary vertex labelling), the output gate will return~$1$ if and only if that graph contains a triangle.

    This circuit is easily checked to be $\Gamma$-symmetric.
    Indeed, each symmetry $\sigma\in \Gamma$ corresponds to the action of a permutation~$\tau \in \Scal_n$ on pairs $\{i,j\} \subseteq [n]$.
    This symmetry straightforwardly extends to an automorphism of the circuit:
    simply take the permutation~$\pi$ which maps each~$v_{\{i, j\}}$ to~$v_{\{\tau(i), \tau(j)\}}$, each $u_{\{i,j,k\}}$ to $u_{\{\tau(i), \tau(j), \tau(k)\}}$, and maps~$out$ to itself.
\end{example}

\subsection{Overview of symmetric computation}
\label{sec:ClassicalResults}

The notion of symmetric circuits introduced above has proven to be robust, naturally capturing the intuition of performing computations in a symmetric manner.
In what follows, we briefly outline some of the key results on such circuits that motivated our investigations.
For a more comprehensive account of this rich area, we refer the reader to Dawar’s survey articles \cite{Dawar2020, Dawar2024}.

Symmetric circuits arose from work in \emph{descriptive complexity theory}, whose main aim is to characterize complexity classes by the type of logic needed to express their corresponding languages.
The search for a logic which captures precisely the languages decidable in polynomial time by a Turing machine (i.e., the class P) has led researchers to consider \emph{Fixed-Point Logic with Counting}, which we denote by FPC.
In short, FPC is an extension of first-order logic with a mechanism for defining predicates inductively and a mechanism for expressing the cardinality of definable sets
(informally, mechanisms for iteration and counting).
It has emerged as a logic of reference in the quest for a logic that captured P;
see Grohe's survey \cite{Grohe2008} and Dawar's survey \cite{Dawar2015} for details.

The logic FPC was first introduced by Immerman \cite{Immerman1986}, who posed the question of whether it completely captures P on the class of finite relational structures (such as graphs).
Even though this question was answered in the negative by Cai, F\"{u}rer and Immerman \cite{CaiFI1992} over three decades ago, FPC -- and some of its stronger variants -- have continued to be the focus of much research.
It proved to be remarkably expressive while at the same time admitting powerful techniques for proving inexpressibility results for it.

The fragment of P corresponding to FPC was later given a simple characterization in terms of Boolean circuits by Anderson and Dawar \cite{AndersonD2017}:
it corresponds precisely to those problems decidable by a uniform family of polynomial-size symmetric threshold circuits!
Inexpressibility results for FPC thus translate into super-polynomial lower bounds for symmetric circuits.
It is also possible to translate the logical inexpressibility arguments into more combinatorial ones, thus bypassing the logical formalism and directly proving precise super-polynomial (and even exponential) lower bounds for the size of symmetric circuits which decide specific problems;
see Dawar's presentation abstract \cite{Dawar2016} for how this can be done.

While symmetric threshold circuits were introduced because of their close relationship to the logic FPC, their definition is natural, robust and interesting in its own right.
They could just as well have been introduced by complexity theorists studying the extent to which \emph{asymmetry} can be used as a resource in computation, or by researchers in combinatorial optimisation studying symmetric linear programs that decide graph properties (see \cite{AtseriasDO2021}).
This notion has not received sufficient attention from the wider complexity theory community due to a language barrier, as results concerning symmetric computation are usually stated and proven in the language of logic, which can often present difficulties.

Despite not being able to efficiently decide all problems in~P, many powerful polynomial-time algorithmic techniques can be efficiently expressed in the setting of symmetric threshold circuits.
For instance, the ellipsoid methods for solving linear programs \cite{AndersonDH2015} and for solving semidefinite programs \cite{DawarW2017} can be efficiently implemented in this framework.
As a consequence, one can solve many non-trivial graph optimisation problems in symmetric polynomial time, such as finding the size of a maximum matching or the capacity of a maximum flow.
Moreover, a result of Grohe \cite{Grohe2012, Grohe2017} shows that any class of graphs defined by excluded minors can be efficiently recognised by symmetric threshold circuits.

On the negative side, it is known that no family of polynomial-size symmetric threshold circuits can decide whether a graph is 3-colourable or whether it contains a Hamiltonian cycle.
Other examples of problems with super-polynomial lower bounds are the constraint satisfaction problems 3-SAT and XOR-SAT \cite{Dawar2016}.

Finally, beyond its connection with logic, symmetric threshold circuits have been shown to be closely connected to the theory of linear programming.
Indeed, it turns out that families of \emph{symmetric linear programs} that decide a property of graphs are equivalent, with at most polynomial blow-up in size, to families of symmetric threshold circuits \cite{AtseriasDO2021}.

\subsection{Symmetric reversible circuits}
\label{sec:reversible}

As a first step towards quantizing the above notion of symmetric circuits, we give an equivalent definition of that same class in terms of \emph{reversible circuits}.
As quantum gates represent unitary operations which can be inverted, reversible classical computation forms a natural bridge between the usual notion of Boolean circuits and that of quantum circuits.
In particular, we will transition from the representation of circuits as directed acyclic graphs into one that is closer to the usual representation of quantum circuits.

Note that any Boolean map $f: \bset{n} \to \bset{}$ can be made reversible by introducing one extra ``workspace'' bit:
just consider the $(n+1)$-bit reversible map
\begin{equation} \label{eq:head}
    (a_1, \dots, a_n,\, b) \mapsto \big(a_1, \dots, a_n,\, b\oplus f(a_1, \dots, a_n)\big),
\end{equation}
where $\oplus$ denotes addition modulo $2$.
The input bits are repeated in the output to make the whole process reversible, and the value of the original function is encoded in the last bit of the output.

Informally, reversible circuits are networks of \emph{wires} that carry bit values to \emph{gates} that perform elementary operations on the bits.
The wires are specified by a set~$W$ of \emph{wire labels}, while gates are reversible functions from $\bset{W}$ to $\bset{W}$.

As in the Boolean setting considered before, the only gates we allow in our circuits will be the $\NOT$ gate and the reversible analogues of threshold and equality gates
(as in equation~\eqref{eq:head} above).

The wire in which a $\NOT$ gate acts is specified in the subscript:
for $h\in W$, we denote by~$\NOT_h$ the map from $\bset{W}$ to $\bset{W}$ that flips the bit at index~$h$ and leaves all others unchanged.
In a similar way, we will use the following notation for threshold-type gates:

\begin{definition}[Reversible threshold-type gates]
    Given a set $S \subset W$, an element $h\in W\setminus S$ and an integer $t\geq 0$, we denote by $\thr^{S, h}_{\geq t}: \bset{W} \to \bset{W}$ the map that takes $(a_w)_{w\in W}$ to $(b_w)_{w\in W}$ where
    $$b_w=
    \begin{cases}
    a_w & \text{if }\, w\neq h, \\
    a_h \oplus \one\Big\{ \sum_{s\in S} a_s \geq t \Big\} & \text{if }\, w=h.
    \end{cases}$$
    Likewise, we denote by $\Eq^{S, h}_{t}$ the map that takes $(a_w)_{w\in W}$ to $(b_w)_{w\in W}$ where
    $$b_w=
    \begin{cases}
    a_w & \text{if }\, w\neq h, \\
    a_h \oplus \one\Big\{ \sum_{s\in S} a_s = t \Big\} & \text{if }\, w=h.
    \end{cases}$$
\end{definition}

\paragraph{Layered reversible circuits.}
In order to consider circuit symmetries, we need to impose some extra structure on the reversible circuits. 
This is done by grouping the gates into \emph{layers} $\Lcal_1, \dots, \Lcal_D$ of pairwise commuting operations.
The layers of a circuit are linearly ordered, so the gates from a prior layer are implemented before those at a later layer, and the wires never feed back to a prior location in the circuit.
Since gates within a single layer~$\Lcal_i$ commute, their application has no intrinsic order.

In summary, the reversible circuits we consider here are formally defined as follows:

\begin{definition}[Layered reversible circuits]
    A layered reversible circuit is specified by a set of wire labels~$W$ and a sequence of layers $(\Lcal_1, \dots, \Lcal_D)$.
    Each layer $\Lcal_i$ is a set of pairwise commuting operations from $\bset{W}$ to $\bset{W}$, where each of these operations is either a $\NOT$ gate (of the form $\NOT_h$ for some $h\in W$) or a reversible threshold-type gate (of the form $\thr^{S,h}_{\geq t}$ or $\Eq^{S,h}_t$ for some $S\subset W$, $h\in W\setminus S$ and $t\geq 0$).
\end{definition}

\begin{remark}
    One could also define reversible circuits without specifying its layers, but only its wire labels and an ordering of its gates.
    In such a case, two circuits whose collection of gates is the same but their ordering differ by permutations within commuting layers are seen to be isomorphic.
    We explicitly give a choice of layers when defining reversible circuits to simplify the presentation when considering circuit symmetries
    (where isomorphic circuits are treated as equivalent).
\end{remark}

\paragraph{Symmetry and automorphisms.}
Given a permutation $\pi\in \Scal_W$, we define its action on sets and gates in the following way:
\begin{itemize}
    \item If $S\subseteq W$ is a set, then $\pi(S) = \{\pi(s): s\in S\}$.
    \item If $h\in W$, then $\pi(\NOT_h) = \NOT_{\pi(h)}$.
    \item If $S\subseteq W$ and $h\in W\setminus S$, then $\pi(\thr^{S, h}_{\geq t}) = \thr^{\pi(S), \pi(h)}_{\geq t}$ and $\pi(\Eq^{S, h}_t) = \Eq^{\pi(S), \pi(h)}_t$.
\end{itemize}
An \emph{automorphism} of a reversible circuit is then a permutation of its wire labels which leaves all of its layers unchanged.
More precisely:

\begin{definition}[Reversible circuit automorphism]
    Let $C = (W, \Lcal_1, \dots, \Lcal_D)$ be a layered reversible circuit.
    An automorphism of~$C$ is a permutation $\pi\in \Scal_W$ such that the following holds:
    for every layer~$\Lcal_i$ and every gate $g\in \Lcal_i$, the function~$\pi(g)$ is also a gate belonging to~$\Lcal_i$.
\end{definition}

\paragraph{Symmetric circuit definition.}
We specify a subset $X\subseteq W$ to label the ``input bits'' for our circuits.
Those are the bits that encode the specific input of the computation, and thus correspond to the structure on which our circuits act (as in the previous section).
All other wires will contain the ``workspace bits'', which are initialised~$0$ and are only present to help in the computation.

Suppose that the structure~$X$ on which a reversible circuit~$C$ acts has an automorphism group $\Gamma\leq \Scal_X$.
We say that~$C$ is \emph{$\Gamma$-symmetric} if every symmetry in~$\Gamma$ can be extended to an automorphism of~$C$.
More precisely:

\begin{definition}[Symmetric reversible circuits]
    Let~$C$ be a layered reversible circuit with wire labels~$W$ and input labels~$X\subseteq W$, and let $\Gamma\leq \Scal_X$ be a group.
    We say that~$C$ is $\Gamma$-symmetric if, for every $\sigma\in \Gamma$, there is an automorphism~$\pi$ of~$C$ such that $\pi(x) = \sigma(x)$ for all $x\in X$.
\end{definition}

As a concrete example of symmetric reversible circuits, we revisit our earlier example (Example~\ref{ex:triangle_bool}) concerning triangle-free graphs:

\begin{example}[Triangle-freeness] \label{ex:triangle_rev}
    We wish to construct a symmetric reversible circuit on $n$-vertex graphs that decides whether some given graph contains a triangle.
    As the desired circuit acts on $n$-vertex graphs, we take~$X$ to be the set of all potential edges in the graph:
    $$X = \big\{ \{i, j\}:\, 1\leq i< j\leq n\big\}.$$
    The symmetries of the problem correspond to pair permutations induced by vertex permutations:
    $$\Gamma = \big\{\{i, j\} \mapsto \{\tau(i), \tau(j)\}:\, \tau \in \Scal_n\big\}.$$
    
    The threshold circuit constructed in Example~\ref{ex:triangle_bool} can be easily made reversible by adding one extra workspace wire for each non-input gate used.
    More precisely, we consider the circuit with wire labels
    $$W := X \cup \big\{ \{i, j, k\}:\, 1\leq i< j < k\leq n\big\} \cup \{out\}$$
    which is composed of two layers, $\Lcal_1$ and $\Lcal_2$.
    The first layer~$\Lcal_1$ is composed of all $\binom{n}{3}$ gates $\thr^{S_{i,j,k}, \{i, j, k\}}_{\geq 3}$ where $1\leq i< j< k\leq n$ and $S_{i,j,k} = \big\{\{i, j\}, \{i, k\}, \{j, k\}\big\}$
    (note that these gates pairwise commute).
    The second layer~$\Lcal_2$ contains a single gate $\thr^{S_{\Delta},\, out}_{\geq 1}$ where
    $$S_{\Delta}= \big\{\{i, j, k\}:\, 1\leq i< j< k\leq n \big\}.$$
    
    At the start, each input wire labelled by a pair $\{i, j\} \in X$ is initialised~$1$ or~$0$ depending on whether or not there is an edge between vertices~$i$ and~$j$;
    all workspace wires are initialised~$0$.
    After applying layer~$\Lcal_1$, the bit corresponding to a wire labelled $\{i, j, k\}$ will be~1 if and only if the vertices labelled by~$i$, $j$ and~$k$ are pairwise connected (that is, if they form a triangle).
    After layer~$\Lcal_2$, the bit corresponding to wire $out$ will be~$1$ if and only if there is a triangle in the graph, which is the property we wished to compute.
    This circuit is $\Gamma$-symmetric:
    every symmetry $\sigma\in \Gamma$ corresponds to the action of a permutation~$\tau \in \Scal_n$ on pairs $\{i,j\} \in X$, and can be extended to a permutation $\sigma'$ of all wire labels by defining
    $$\sigma'(\{i,j,k\}) = \{\tau(i), \tau(j), \tau(k)\} \quad \text{and} \quad \sigma'(out) = out.$$
    Since the action of~$\sigma'$ only permutes the gates within layer~$\Lcal_1$ (and leaves~$\Lcal_2$ invariant), it corresponds to an automorphism of the circuit.
\end{example}

\paragraph{Equivalence of symmetric threshold and reversible circuits.}
Even though the definition of symmetric reversible circuits might seem very different from the earlier notion of symmetric threshold circuits, one can show that they are \emph{equivalent}
up to a constant factor on the number of gates used.
This is given by the following proposition, whose proof will be provided in Appendix~\ref{sec:appendix}.

\begin{prop}[Equivalence of definitions]
\label{prop:equiv}
    Let~$X$ be a finite set and $\Gamma\leq \Scal_X$ be a group.
    Then:
    \begin{itemize}
        \item[$(i)$] Any $\Gamma$-symmetric threshold circuit with~$s$ gates can be converted into an equivalent $\Gamma$-symmetric reversible circuit which uses at most~$2s$ gates and~$s$ workspace bits.
        \item[$(ii)$] Any $\Gamma$-symmetric reversible circuit with~$s$ gates can be converted into an equivalent $\Gamma$-symmetric threshold circuit with at most~$4s$ (non-input) gates.
    \end{itemize}
\end{prop}

\begin{remark}
    The number of workspace bits given in this proposition is generally sub-optimal, and may be reduced by cleverly reusing the workspace bits as shown by Bennett \cite{Bennett1973, Bennett1989}.
    We will not discuss such improvements in this work.
\end{remark}

Thus, all the results concerning symmetric threshold circuits given in the last subsection translate without changes to our notion of layered reversible circuits.

\paragraph{Power and robustness of the gate set.}
Finally, we show that our notion of symmetric reversible circuits is sufficiently powerful to efficiently compute any reversible function possessing a rich enough symmetry group, which we shall call \emph{partition symmetry}.
This notion of symmetry is motivated by the fact that reversible threshold functions $\thr^{S, h}_{\geq t}$ are not fully symmetric (i.e., permutation-invariant), as the role of the ``head''~$h$ differs from the elements in the support~$S$;
however, within each of these two classes, their action is fully symmetric.
A generalisation of this type of symmetry gives rise to the following natural notion:

\begin{definition}[Partition-symmetric functions]
    A function $F: \prod_{i=1}^k \bset{n_i} \to \prod_{i=1}^k \bset{n_i}$ is said to be partition-symmetric if its action commutes with the natural action of the group $\prod_{i=1}^k \Scal_{n_i}$.
\end{definition}

In Appendix~\ref{sec:appendix} we prove the following result:

\begin{theorem}[Computing partition-symmetric functions] \label{thm:part_sym}
    Let $k\geq 1$ and $n_1, \dots, n_k$ be integers, and denote $n = n_1 +\dots + n_k$.
    Any partition-symmetric reversible function
    $F: \prod_{i=1}^k \bset{n_i} \to \prod_{i=1}^k \bset{n_i}$
    can be implemented by a partition-symmetric reversible circuit using $O\big(n\prod_{i=1}^k (n_i+1)\big)$ gates and $O(n)$ workspace bits, all of which are returned to zero at the end of the computation.
\end{theorem}

Other than serving as a testament to the power of symmetric circuits, this result also shows that their definition of is significantly \emph{robust}:
if we were to add any number of partition-symmetric functions (with a small number of parts) to our allowed gate set, this would only modestly change the complexity of computing any given function.
In particular, the class of polynomial-time symmetrically computable functions is unaffected by such alterations in the gate set.

\medskip

We are finally ready to introduce the notion of symmetric \emph{quantum} circuits, which emerges quite naturally from the definition of symmetric reversible circuits.

\subsection{Symmetric quantum circuits}
\label{sec:QuCircuits}

A quantum circuit is composed of a set of \emph{qubit wires} (labelled by a finite set~$W$) and a collection of \emph{quantum gates}.
Each wire $w\in W$ is associated with a two-dimensional complex Hilbert space~$\Hcal_w$ equipped with a fixed orthonormal basis denoted $\{\ket{0}_w,\, \ket{1}_w\}$, called the \emph{computational basis}.
The Hilbert spaces associated with the qubit wires are pairwise orthogonal, and the \emph{state space} of the wire set~$W$ is the tensor product of the Hilbert spaces of all constituent wires:
$\Hcal_W = \bigotimes_{w\in W} \Hcal_w$.
The computational basis of the circuit is formed by tensoring the computational bases of each wire:
it is given by
$$\bigg\{ \bigotimes_{w\in W} \ket{a_w}_w:\: a_w \in \{0, 1\} \text{ for all } w\in W\bigg\}.$$
We omit the indices when they are clear from context.

Quantum gates are unitary maps acting on the state space~$\Hcal_W$.
As before, we restrict our attention to gates which are either of threshold-type or which act on a single qubit wire.
The definition of quantum threshold gates is given as follows:

\begin{definition}[Quantum threshold-type gates]
    Let~$W$ be a set of wire labels.
    Given a set $S \subset W$, an element $h\in W\setminus S$ and an integer $t\geq 0$, we denote by $\thr^{S, h}_{\geq t}: \Hcal_W \to \Hcal_W$ the linear map that takes a computational basis state $\bigotimes_{w\in W} \ket{a_w}_w$ to $\bigotimes_{w\in W} \ket{b_w}_w$ where
    $$b_w=
    \begin{cases}
    a_w & \text{if }\, w\neq h, \\
    a_h \oplus \one\Big\{ \sum_{s\in S} a_s \geq t \Big\} & \text{if }\, w=h.
    \end{cases}$$
    Likewise, $\Eq^{S, h}_{t}: \Hcal_W \to \Hcal_W$ denotes the linear map that takes $\bigotimes_{w\in W} \ket{a_w}_w$ to $\bigotimes_{w\in W} \ket{b_w}_w$ where
    $$b_w=
    \begin{cases}
    a_w & \text{if }\, w\neq h, \\
    a_h \oplus \one\Big\{ \sum_{s\in S} a_s = t \Big\} & \text{if }\, w=h.
    \end{cases}$$
\end{definition}

This definition is readily checked to give unitary maps.
Note that both the CNOT and the Toffoli gates are examples of quantum threshold gates.

\begin{remark}
One can easily show that the equality gates $\Eq^{S, h}_{t}$ are unnecessary for the definition of symmetric quantum circuits, as they can be efficiently (and symmetrically) implemented using threshold gates.
We include them in the gate set for convenience when describing explicit circuits.
\end{remark}

We also allow in our circuits any unitary that acts only on a single qubit.
This can be seen as a natural quantum extension of the gate set we allowed in reversible circuits, as the $\NOT$ gate is the only non-trivial single-bit reversible map.
Important examples of single-qubit unitaries which will be used later on are the Hadamard matrix~$H$, the Pauli matrices $X$, $Y$ and $Z$, and the parametrised \emph{phase gates}~$P_{\alpha}$ and \emph{rotation gates} $R_y(\theta)$ defined by
\begin{equation*}
    P_{\alpha} :=
    \begin{pmatrix}
        1 & 0\\
        0 & e^{i\alpha}
    \end{pmatrix},
    \quad R_y(\theta) := e^{i\theta Y/2} =
    \begin{pmatrix}
        \cos(\theta/2) & -\sin(\theta/2)\\
        \sin(\theta/2) & \cos(\theta/2)
    \end{pmatrix}.
\end{equation*}
We denote the set of all $2\times 2$ unitary matrices by $\U(2)$.
Given a unitary $U\in \U(2)$ and a label $w\in W$, we denote by~$U_w$ the unitary map on~$\Hcal_W$ that is the tensor product of~$U$ acting on wire~$w$ with the identity on every other wire.
(More precisely, this action is determined by the matrix~$U$ when we assume~$\ket{0}_w$ is the vector indexing the first row/column of~$U$ and~$\ket{1}_w$ indexes the second row/column.)

\paragraph{Notation for quantum circuits.}
Given a quantum circuit~$\Ccal$, we denote by $s(\Ccal)$ its total number of gates and by $a(\Ccal)$ its number of workspace qubits.\footnote{The~$s$ is for \emph{size} and the~$a$ is for \emph{ancilla}.}
For technical reasons, we will sometimes need to consider the number $h(\Ccal)$ of qubit wires acting as the ``head''~$h$ of some quantum threshold gate $\thr^{S,h}_{\geq t}$ or $\Eq^{S,h}_{t}$.

As in the classical reversible setting, to incorporate symmetries into our circuits we will need to group gates into \emph{layers} of pairwise-commuting operations.
The intuition is that adjacent commuting gates in a quantum circuit do not have a well-defined order of implementation, and thus exchanging their order should make no difference to the circuit.
We arrive at the following notion.

\begin{definition}[Layered quantum circuits]
    A layered quantum circuit is specified by a set of wires~$W$ and a sequence $(\Lcal_1, \dots, \Lcal_D)$ of layers, where:
    \begin{itemize}
        \item Each wire $w\in W$ is associated with a two-dimensional complex Hilbert space~$\Hcal_w$ equipped with a fixed orthonormal basis $\{\ket{0}_w,\, \ket{1}_w\}$;
        the state space of the circuit is $\Hcal_W = \bigotimes_{w\in W} \Hcal_w$.
        \item Each layer $\Lcal_i$ is a set of pairwise commuting quantum gates, where each of these gates is either a single-qubit gate (of the form $U_h$ for some $U\in \U(2)$ and $h\in W$) or a quantum threshold gate (of the form $\thr^{S,h}_{\geq t}$ or $\Eq^{S,h}_t$ for some $S\subset W$, $h\in W\setminus S$ and $t\geq 0$).
    \end{itemize}
\end{definition}

Note that a given quantum circuit may be expressed in many different ways as a layered quantum circuit;
likewise, any ordering of gates within the layers gives rise to a different (but equivalent) sequential implementation of the layered quantum circuit.
This freedom in reordering gates without changing the circuit is essential for the definition of symmetry in quantum circuits.

\paragraph{Symmetry and automorphisms.}
A permutation $\pi\in \Scal_W$ acts on subsets of~$W$ element-wise, and it acts on gates in the following way:
\begin{itemize}
    \item If $h\in W$ and $U\in \U(2)$, then $\pi(U_h) = U_{\pi(h)}$.
    \item If $S\subseteq W$ and $h\in W\setminus S$, then $\pi(\thr^{S, h}_{\geq t}) = \thr^{\pi(S), \pi(h)}_{\geq t}$ and $\pi(\Eq^{S, h}_t) = \Eq^{\pi(S), \pi(h)}_t$.
\end{itemize}
The permutation~$\pi$ also acts on a layer~$\Lcal$ element-wise, i.e., by acting individually on each of its gates.

\begin{remark}
    A common, but more abstract, way of defining symmetric actions on a quantum circuit is through the natural representation~$R(\cdot)$ of the symmetric group~$\Scal_W$ on the state space~$\Hcal_W$ of the circuit;
    this representation is determined by
    \begin{equation} \label{eq:representation}
        R(\pi) \bigotimes_{w\in W} \ket{a_w}_w = \bigotimes_{w\in W} \ket{a_{\pi^{-1}(w)}}_w \quad \text{for $a\in \bset{W}$.}
    \end{equation}
    Using this notation, the action of $\pi\in \Scal_W$ on a gate~$g$ as given above can be expressed as $\pi(g) = R(\pi) g R(\pi)^{-1}$, which is the commonly-used notion of permutation action on quantum gates.
\end{remark}

An \emph{automorphism} of a layered quantum circuit is a permutation of its wires which leaves the circuit unchanged, in the sense that it maps every layer to itself.

\begin{definition}[Quantum circuit automorphism] \label{def:automorphism}
    Let~$\Ccal$ be a layered quantum circuit with wire labels~$W$ and layers $\Lcal_1, \dots, \Lcal_D$.
    An automorphism of~$\Ccal$ is a permutation $\pi\in \Scal_W$ which induces a permutation of each layer~$\Lcal_i$:
    $\pi(g) \in \Lcal_i$ for all $g\in \Lcal_i$ and all $i\in [D]$.
\end{definition}

\paragraph{Symmetric circuit definition.}
Our quantum circuits will act on some specified structure~$X$, which is assumed to label a subset of the qubits
(as in Example~\ref{ex:triangle_rev} from the last section).
Those qubits labelled by~$X$ are called ``active qubits'', while the remaining are called ``workspace qubits''.
Given a group of symmetries $\Gamma$ acting on the structure~$X$, we say that a quantum circuit~$\Ccal$ is $\Gamma$-symmetric if every symmetry in~$\Gamma$ can be extended to an automorphism of~$\Ccal$.

\begin{definition}[Symmetric quantum circuits]
    Let~$\Ccal$ be a quantum circuit with wire labels~$W$ and an ordered sequence of gates $(g_1, \dots, g_m)$, let~$X\subseteq W$ be a subset representing the active qubits and let $\Gamma\leq \Scal_X$ be a group.
    We say that~$\Ccal$ is $\Gamma$-symmetric if its gates can be grouped into layers of pairwise-commuting operations so that the following holds:
    for every $\sigma\in \Gamma$, there is an automorphism~$\pi$ of~$\Ccal$ such that $\pi(x) = \sigma(x)$ for all $x\in X$.
\end{definition}

\paragraph{The input and output of a quantum circuit.}
The definition of symmetric quantum circuits given above makes no reference to the input or output of the computation.
The reason for this is that, depending on the task at hand, there are different ways in which to initialise the state at the beginning of the circuit and different ways in which to read out the output at the end of the computation.
It is more convenient to specify these details when considering each specific problem;
below we give some examples and general rules.

The first general rule is that \emph{workspace qubits are always initialised~$\ket{0}$}.
This is done because we do not wish workspace qubits to encode any extra information pertaining to the computation.

When dealing with classical problems (i.e., those whose input and output are both bit strings), the input qubits are those labelled by~$X$ and are initialised on a computational basis state encoding the specific input to the problem.
If this is a decision problem, then the output will be encoded into a designated workspace qubit and may be read by measuring it in the computational basis.

For state preparation circuits, all qubits are initialised~$\ket{0}$ and the desired state~$\ket{\psi}$ will be encoded in the active qubits (those labelled by~$X$), while all workspace qubits are required to end at state~$\ket{0}$.
It is then common in quantum computing to disregard the workspace qubits and say that the circuit prepares the state~$\ket{\psi}$.

We will also consider circuits that implement algorithmic subroutines, which might not have well-defined inputs or outputs but will instead depend on the larger algorithm they are constituents of.

\paragraph{The role of measurements.}
Our notion of symmetric quantum circuits does not allow intermediate measurements, and one might wonder whether this unnecessarily restricts their power.
Indeed, if one were to allow (say) two-qubit entangled measurements, then these measurements would be able to break the symmetries of the circuit and might lead to a richer class of final quantum states at the end of the computation.
However, such a breaking of symmetries is precisely what we wish to avoid in our computational framework.

A simple way to remedy this is to allow only single-qubit intermediate measurements, for instance by adding such measurement operators to our allowed gate set.
It is easy to show that this possibility does not give any extra power to symmetric circuits, due to a ``symmetric principle of deferred measurement''.
More precisely, we can simulate measurements in the computational basis by using an extra workspace qubit for each qubit measured, this extra qubit being entangled with the measured one by way of a CNOT gate.
If we do this for all measurements (always using new workspace qubits) and then erase these measurements (in particular replacing subsequent classical controlled gates by quantum controlled gates), the new circuit will have the same symmetries as the original one and the same output distribution, but now all measurements are performed at the end.
General single-qubit measurements can be performed similarly, using extra single-qubit unitary gates before and after the CNOT.

One may then allow for intermediate single-qubit measurements without significant alterations to the class of symmetric circuits obtained.
We have chosen not to include such measurements in our definition for simplicity of exposition.

\paragraph{The complexity class symBQP.}
One can also define computational classes of problems that can be solved efficiently by symmetric quantum circuits.
More precisely, let $(X_n)_{n}$ be a sequence of structures, each one having an intrinsic symmetry group~$\Gamma_n$.
A language $L\subseteq \bigcup_n \bset{X_n}$ is in complexity class symBQP/poly if the following holds:
For every~$n$, there exists a poly$(|X_n|)$-size $\Gamma_n$-symmetric quantum circuit~$\Ccal_n$ such that, if one initialises the active qubits with~$x_n \in \bset{X_n}$, then the output qubit when measured at the end of the computation will be~$\one[x\in L]$ with probability at least~$2/3$.
A language $L\subseteq \bigcup_n \bset{X_n}$ is in complexity class symBQP if it is in symBQP/poly and there exists a polynomial-time Turing machine which, on input~$1^{|X_n|}$, outputs the description of a $\Gamma_n$-symmetric quantum circuit~$\Ccal_n$ as above.

These classes comprise all efficiently solvable symmetric problems and represent rich and interesting objects of study.
They are quantum analogues of the corresponding classical symmetric complexity classes (which have tight connections to logic~\cite{Dawar2024}) and merit further study.

\section{Building blocks}
\label{sec:build_block}

To facilitate the construction of symmetric quantum circuits in later sections, we start by providing some useful \emph{building blocks}.
For the rest of this section, let~$X$ be a finite set corresponding to the structure on which our circuits act, let $\Gamma \leq \Scal_X$ be its group of symmetries and denote $n := |X|$.

The most important building block in the construction of symmetric circuits is the operation of \emph{circuit concatenation}, as it allows us to construct large symmetric circuits by gluing together several smaller ones.
While this operation can be made to preserve the symmetries of the original circuits, one must be careful in how to concatenate their workspace qubits;
this is illustrated as follows.

Let~$\Ccal_1$ and~$\Ccal_2$ be $\Gamma$-symmetric circuits over~$X$, and suppose that they use $a(\Ccal_1)$ and $a(\Ccal_2)$ workspace qubits respectively.
There are two natural ways in which we can implement their concatenation:
either by using the same set of $\max\{a(\Ccal_1), a(\Ccal_2)\}$ workspace qubits for both circuits, or by having $a(\Ccal_1)+a(\Ccal_2)$ workspace qubits in total and using disjoint subsets for each of~$\Ccal_1$ and~$\Ccal_2$.
These possibilities are shown in Figure~\ref{fig:concat} below.

\begin{figure}[ht]
    \centering
    \includegraphics[width=15cm]{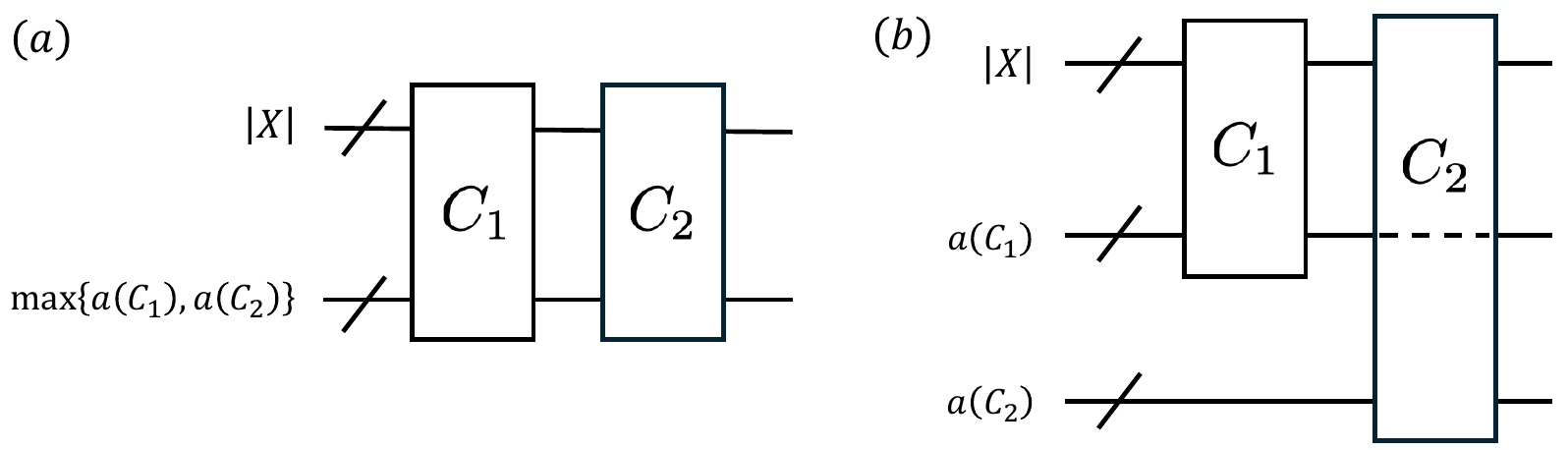}
    \caption{Two ways of concatenating circuits which use workspace qubits.}
    \label{fig:concat}
\end{figure}

One can see that the second form of implementation (Figure~\ref{fig:concat}(b)) will always be $\Gamma$-symmetric.
However, it is not necessarily the case that the first way of concatenating the circuits is $\Gamma$-symmetric, as we might need to permute the workspace qubits of each circuit in a different way when the active qubits are acted upon by a symmetry in~$\Gamma$.
It can only be guaranteed that the circuit in Figure~\ref{fig:concat}(a) will be $\Gamma$-symmetric if the workspace qubits of~$\Ccal_1$ and~$\Ccal_2$ have the same symmetries.

For ease of reference, we state the conclusion of the discussion above as the following lemma:

\begin{lemma}[Circuit concatenation]
\label{lem:concat}
    Suppose~$\Ccal_1$ and~$\Ccal_2$ are $\Gamma$-symmetric circuits.
    Then:
    \begin{itemize}
        \item[$(i)$] The concatenation of~$\Ccal_1$ and~$\Ccal_2$ where we use \emph{disjoint workspace qubits} for each circuit is also $\Gamma$-symmetric.
        \item[$(ii)$] If the workspace qubits of~$\Ccal_1$ and~$\Ccal_2$ have the same symmetries, then the concatenated circuit $\Ccal_1 \circ \Ccal_2$ using \emph{the same workspace qubits} is also $\Gamma$-symmetric.
    \end{itemize}
\end{lemma}

Other than circuit concatenation, a fundamental operation that is easily seen to preserve symmetry is that of \emph{circuit inversion}:

\begin{lemma}[Circuit inversion]
\label{lem:inversion}
    Given a $\Gamma$-symmetric circuit~$\Ccal$ which implements a unitary~$U$, we can construct a $\Gamma$-symmetric circuit~$\Ccal^{-1}$ which implements~$U^{-1}$ and has the same complexity as~$\Ccal$.
\end{lemma}

\begin{proof}
    Since threshold gates are their own inverse, the circuit~$\Ccal^{-1}$ formed by replacing every single-qubit gate~$A$ of~$\Ccal$ by its inverse~$A^{\dagger}$ and then reversing the order of layers (with respect to $\Ccal$) will implement~$U^{-1}$.
    It is clear that~$\Ccal^{-1}$ will have the same symmetries and the same complexity as~$\Ccal$.
\end{proof}

The next lemma shows that it is also possible to perform controlled operations while preserving the symmetries of the original circuit, incurring in only a linear increase in circuit complexity.
Recall that~$h(\Ccal)$ denotes the number of wires which serve as the head of some threshold gate in circuit~$\Ccal$.

\begin{lemma}[Controlled operations] \label{lem:control}
    Given a $\Gamma$-symmetric circuit~$\Ccal$, we can construct an $(\{id_1\} \times\Gamma)$-symmetric circuit~c-$\Ccal$ which implements control-$\Ccal$ using $O(s(\Ccal))$ gates and $a(\Ccal)+h(\Ccal)$ workspace qubits.
\end{lemma}

\begin{proof}
    We modify the original circuit one layer at a time, transforming each layer of $\Ccal$ into (at most) eight layers of c-$\Ccal$.
    Other than the original qubits of~$\Ccal$ and the control qubit, the controlled circuit will have an extra workspace qubit~$q'$ for each qubit~$q$ in~$\Ccal$ which is the head of some threshold gate.

    Start at the first layer~$\Lcal_1$ of~$\Ccal$, and decompose it into a layer~$\Lcal_1'$ containing only single-qubit gates and a layer~$\Lcal_1''$ containing only threshold gates.
    We can assume that there is at most one single-qubit gate in~$\Lcal_1'$ acting on any given qubit
    (otherwise we can first multiply them together).
    
    Let $U\in \U(2)$ be the gate in~$\Lcal_1'$ acting on a given qubit.
    By \cite[Corollary~4.2]{NielsenChuang2010}, there exist unitaries $A$, $B$, $C\in \U(2)$ and an angle $\alpha$ such that
    $$ABC = I \quad \text{and} \quad U = e^{i\alpha} AXBXC.$$
    We may then implement control-$U$ using four single-qubit gates and two CNOTs distributed across five layers, as shown in Figure~\ref{fig:single_qubit}.
    Doing so for all gates in~$\Lcal_1'$ and combining the resulting gates into the same five layers, we obtain five $\Gamma$-symmetric layers which together implement c-$\Lcal_1'$.

    \begin{figure}[ht]
        \centering
        \includegraphics[width=12cm]{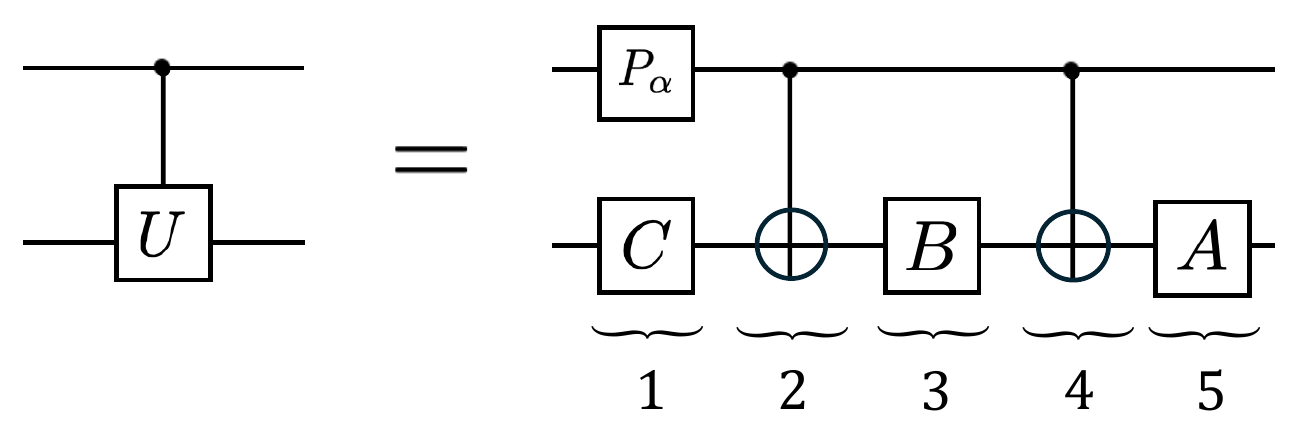}
        \caption{Implementing a controlled single-qubit gate in five layers.}
        \label{fig:single_qubit}
    \end{figure}

    Passing to~$\Lcal_1''$, each threshold gate is replaced by three threshold gates distributed across three layers of c-$\Ccal$ as shown in Figure~\ref{fig:threshold}.
    We use a different workspace qubit~$h'$ for each qubit~$h$ which is the head of a threshold gate, and group the resulting gates into the same three layers.
    Note that the resulting layers will each contain only pairwise-commuting gates, and that together they implement c-$\Lcal_1''$.
    (This is not quite trivial, but follows easily from equation~\eqref{eq:commuting_thr} in Appendix~\ref{sec:appendix}.)
    
    \begin{figure}[ht]
        \centering
        \includegraphics[width=14cm]{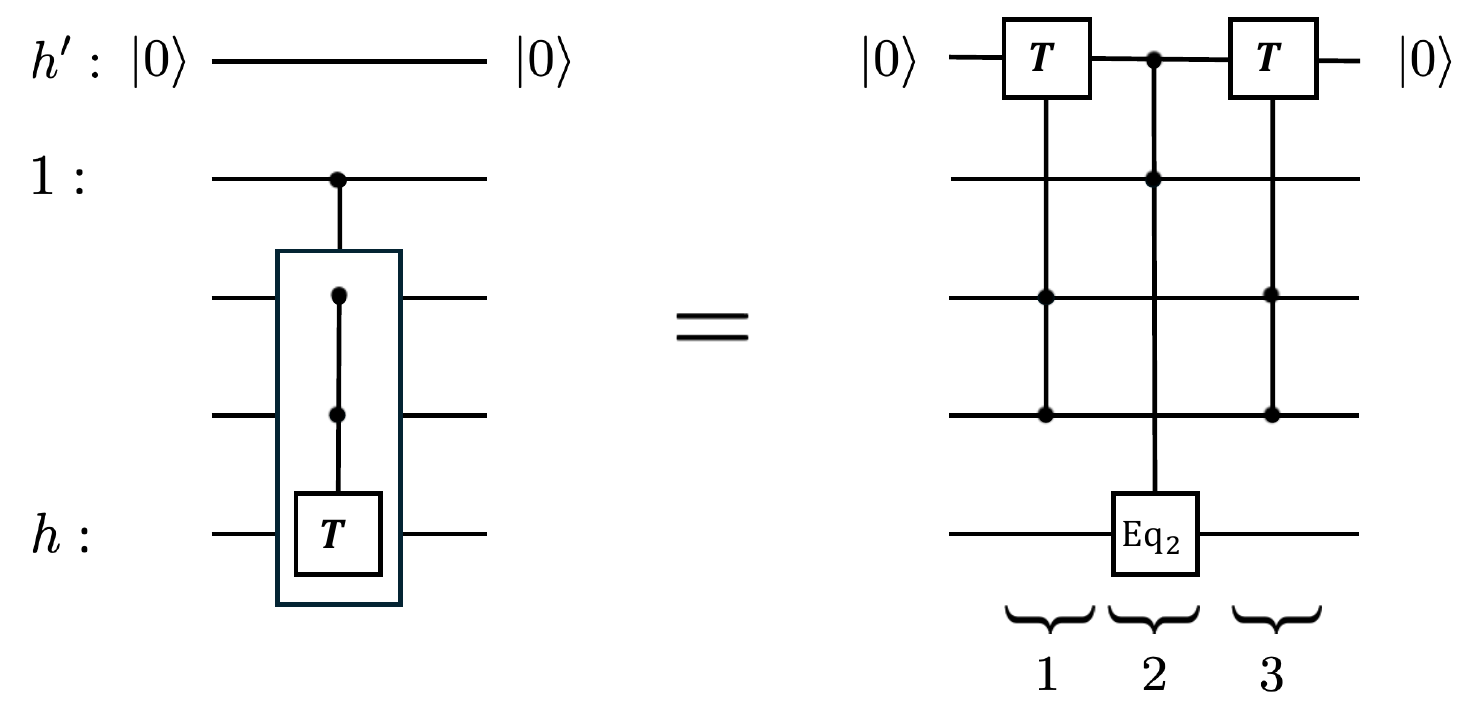}
        \caption{Implementing a controlled threshold gate in three layers.}
        \label{fig:threshold}
    \end{figure}

    Doing the same for each subsequent layer of~$\Ccal$ and then concatenating the resulting layers, by Lemma~\ref{lem:concat} we obtain a $\Gamma$-symmetric circuit implementing control-$\Ccal$.
    Note that we can reuse the same workspace qubits throughout the circuit c-$\Ccal$, as (by construction) those qubits respect the same symmetries within each layer.
\end{proof}

Another important building block for quantum circuits is the implementation of \emph{select operations} $\sum_j \ket{j}\!\bra{j} \otimes U_j$, where one of several unitaries $U_0, \dots, U_{k-1}$ is selected to be performed depending on the state of some chosen qubits.
The following lemma shows how to implement these select operations in a way that preserves the symmetries of the unitaries to be selected.

\begin{lemma}[Select operations]
\label{lem:select_op}
    Let $U_0, \dots, U_{k-1}$ be $\Gamma$-symmetric $n$-qubit unitaries, let $\theta_0, \dots, \theta_{k-1}$ be real numbers and denote $m := \lceil \log k \rceil$.
    Given $\Gamma$-symmetric circuits $\Ccal_0, \dots, \Ccal_{k-1}$ such that~$\Ccal_j$ implements~$U_j$ for $0\leq j < k$, we can construct an $\{id_m\} \times\Gamma$-symmetric circuit~$\Ccal_{sel}$ which implements
    \begin{equation*}
        \ket{j} \ket{\psi} \mapsto
        e^{i \theta_j} \ket{j} U_j\ket{\psi} \quad \text{for $0\leq j < k$}
    \end{equation*}
    using $O\big(k\log k + \sum_{j=0}^{k-1} s(\Ccal_j)\big)$ gates and $O\big(n + \sum_{j=0}^{k-1} a(\Ccal_j)\big)$ workspace qubits.
\end{lemma}

\begin{proof}
    The action of the circuit $\Ccal_{sel}$ we wish to implement can be written in the following way:
    \begin{equation} \label{eq:select}
        \prod_{j=0}^{k-1} \big( \ket{j}\!\bra{j} \otimes e^{i \theta_j} U_j + (I_{2^m} - \ket{j}\!\bra{j}) \otimes I_{2^n} \big).
    \end{equation}
    (The terms in this product commute, so the order of multiplication is irrelevant.)
    Each term in the product above can be implemented (with the help of workspace qubits) as shown in Figure~\ref{fig:select}.
    
    \begin{figure}[ht]
        \centering
        \includegraphics[width=14cm]{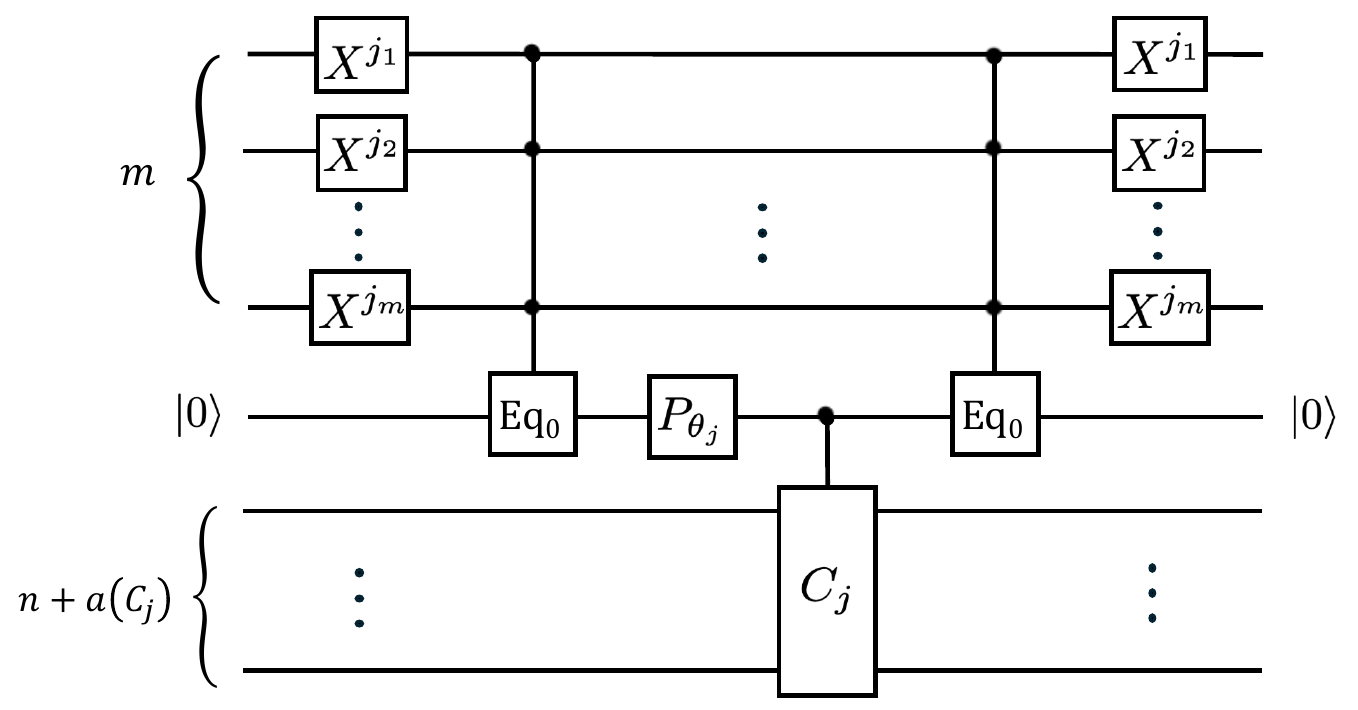}
        \caption{Implementing the circuit $\ket{j}\!\bra{j} \otimes e^{i \theta_j} U_j + (I_{2^m} - \ket{j}\!\bra{j}) \otimes I_{2^{n}}$ using one extra ancilla and a controlled $\Ccal_j$.}
        \label{fig:select}
    \end{figure}

    By Lemma~\ref{lem:control}, we can implement control-$\Ccal_j$ symmetrically using (at most) $n + 2a(\Ccal_j)$ ancillas and $O(s(\Ccal_j))$ gates.
    Note that the~$n$ ancillas coming from `doubling' the original active qubits can be reused in every controlled circuit c-$\Ccal_j$ to be implemented, as they respect the same symmetries
    (which ultimately come from~$\Gamma$).
    By Lemma~\ref{lem:concat} and the formula~\eqref{eq:select} above, we can implement the desired select operation by concatenating all such circuits for $0\leq j\leq k-1$.
    The result follows.
\end{proof}

\begin{remark}
    If the workspace qubits of each circuit~$\Ccal_j$ have the same symmetries, then we can decrease the number of workspace qubits used in the select circuit above by reusing the same qubits as in Figure~\ref{fig:concat}(a).
\end{remark}

Finally, we prove a technical result regarding how arbitrary permutation-symmetric unitaries decompose as a linear combination of simple symmetric unitaries.
This will allow us to prove a quantum analogue of Theorem~\ref{thm:part_sym}, stating that we can implement any partition-symmetric unitary while respecting its symmetries
(see Theorem~\ref{thm:qu_part_sym}).

\begin{lemma}[Decomposition of permutation-symmetric unitaries]
    Let $n\geq 1$ be an integer and denote $m = \binom{n+3}{3}$.
    There exist unitaries $V_1, \dots, V_m \in \U(2)$ and a constant $C(n)>0$ such that the following holds:
    any $\Scal_n$-symmetric $n$-qubit unitary $U$ can be written in the form
    $$U = \sum_{i = 1}^m \alpha_i V_i^{\otimes n}$$
    with $\sum_{i = 1}^m |\alpha_i| \leq C(n)$.
\end{lemma}

\begin{proof}
    Let~$R(\cdot)$ denote the natural representation of the symmetric group~$\Scal_n$ on~$(\C^2)^{\otimes n}$, as given in equation~\eqref{eq:representation}.
    Denote by $(\C^{2\times 2})^{\otimes n}_{sym}$ the vector space of $2^n \times 2^n$ matrices which are invariant under conjugation by all permutation operators $R(\pi)$:
    $$(\C^{2\times 2})^{\otimes n}_{sym} := \big\{A\in (\C^{2\times 2})^{\otimes n}:\: R(\pi) A R(\pi)^\dagger = A \text{ for all } \pi\in \Scal_n \big\}.$$
    It was shown by Watrous \cite[Theorem 7.11]{Watrous2018} that
    $$(\C^{2\times 2})^{\otimes n}_{sym} = \vspan \{V^{\otimes n}:\: V\in \U(2)\},$$
    and that this vector space has dimension $m := \binom{n+3}{3}$
    (see Corollary 7.3 and Proposition 7.10 in \cite{Watrous2018}).
    It follows that there exist single-qubit unitaries $V_1, \dots, V_m \in \U(2)$ such that
    $$\vspan \{V_i^{\otimes n}:\: 1\leq i\leq m\} = \vspan \{V^{\otimes n}:\: V\in \U(2)\} = (\C^{2\times 2})^{\otimes n}_{sym}.$$
    Fix such a set of unitaries $V_1, \dots, V_m$.

    Now consider any $\Scal_n$-symmetric $n$-qubit unitary $U$.
    Since $U\in (\C^{2\times 2})^{\otimes n}_{sym}$, it can be written uniquely as a linear combination of $V_1^{\otimes n}, \dots, V_m^{\otimes n}$.
    To determine the coefficients in this linear combination it suffices to solve a linear system of equations, which we describe next.
    
    Denote by $\langle \cdot, \cdot\rangle$ the normalised trace inner product on $(\C^{2\times 2})^{\otimes n}$:
    $$\langle A, B\rangle = \frac{\tr(B^\dagger A)}{2^n} \quad \text{for } A, B \in (\C^{2\times 2})^{\otimes n}.$$
    By linear independence of the $V_i^{\otimes n}$, the coefficients $(\alpha_1, \dots, \alpha_m) \in \C^m$ in the decomposition $U = \sum_{i = 1}^m \alpha_i V_i^{\otimes n}$ are uniquely determined by the system of linear equations
    $$\langle U, V_i^{\otimes n}\rangle = \bigg\langle \sum_{j=1}^m \alpha_j V_j^{\otimes n},\, V_i^{\otimes n} \bigg\rangle = \sum_{j=1}^m \alpha_j \langle V_j^{\otimes n}, V_i^{\otimes n}\rangle \quad \text{for } 1\leq i \leq m.$$
    Indeed, write $M := \big(\langle V_j^{\otimes n}, V_i^{\otimes n}\rangle \big)_{i, j=1}^m$ and $u := \big(\langle U, V_i^{\otimes n}\rangle \big)_{i=1}^m$.
    Note that $M$ is positive \emph{definite}, since it is the Gram matrix of a set of linearly independent vectors;
    it follows that it is invertible and the operator norm of $M^{-1}$ is equal to $\lambda_{\min}(M)^{-1}$
    (where $\lambda_{\min}(M)>0$ is the smallest eigenvalue of $M$).
    The solution of the linear system above is $\alpha = M^{-1} u$, from which we conclude that
    $$\|\alpha\|_2 \leq \|M^{-1}\|\cdot \|u\|_2 = \lambda_{\min}(M)^{-1} \|u\|_2.$$
    
    By Cauchy-Schwarz we have that $|\langle U, V_i^{\otimes n}\rangle| \leq 1$ for all $i\in [m]$, so $\|u\|_2 \leq \sqrt{m}$, and (also by Cauchy-Schwarz) $\|\alpha\|_1 \leq \sqrt{m} \|\alpha\|_2$.
    It follows that $\|\alpha\|_1 \leq m/\lambda_{\min}(M)$, a bound that depends only on $n$ and our choice of unitaries $V_1, \dots, V_m$.
\end{proof}

\begin{remark}
    In principle it should be possible to obtain a good lower bound for the smallest eigenvalue of such a Gram matrix~$M$, which would translate to a good upper bound for $\|\alpha\|_1$ (and thus for the constant $C(n)$ in the statement).
    We have not been able to do so, but leave it as an open problem in Section~\ref{sec:Open}.
    Underlying this difficulty is the surprising imbalance on the amount of results known for symmetric \emph{states} versus symmetric \emph{unitaries}.
\end{remark}

Applying the previous lemma multiple times, one easily obtains an analogous result concerning partition-symmetric unitaries:

\begin{corollary}[Decomposition of partition-symmetric unitaries]
\label{cor:part_sym}
    Let $n_1,\dots, n_t$ be positive integers with $n_1+\dots+n_t = n$, let $\Gamma = \Scal_{n_1}\times \dots \times \Scal_{n_t}$ and denote $m := \max_{i\leq t} \binom{n_i+3}{3}$.
    There exist unitaries $V_1, \dots, V_m \in \U(2)$ and a constant $C(n)>0$ such that the following holds:
    any $\Gamma$-symmetric $n$-qubit unitary $U$ can be written in the form
    $$U = \sum_{i_1, \dots, i_t = 1}^m \alpha_{i_1, \dots, i_t} V_{i_1}^{\otimes n_1} \otimes \dots \otimes V_{i_t}^{\otimes n_t}$$
    with $\sum_{i_1, \dots, i_t = 1}^m |\alpha_{i_1, \dots, i_t}| \leq C(n)$.
\end{corollary}

\section{Symmetric quantum subroutines}
\label{sec:subroutines}

In this section we show how to perform a number of important quantum subroutines in a symmetric way.

\subsection{Amplitude amplification}

Amplitude amplification is a widely applicable procedure proposed by Brassard, Høyer, Mosca and Tapp \cite{BrassardHMT2000}, which informally allows us to amplify the ``good part'' of the outcome of a quantum algorithm.

The general setup is the following:
suppose we have a quantum circuit~$\Acal$ which prepares a superposition $\sqrt{p} \ket{\psi_1} + \sqrt{1-p} \ket{\psi_0}$, and we wish to increase the amplitude of the state~$\ket{\psi_1}$ in this superposition.
Amplitude amplification gives a general technique for accomplishing this task, assuming only that we can ``distinguish'' between the states~$\ket{\psi_1}$ and~$\ket{\psi_0}$, and that we know an approximation of the original amplitude~$\sqrt{p}$.

The following lemma (which gives the first part of Theorem~\ref{thm:subroutines}) states that we can perform this procedure in a way that respects the symmetries of the original algorithm~$\Acal$ and of the state distinguisher.
For simplicity, it assumes that we know the exact value of the amplitude~$\sqrt{p}$, and in return allows us to exactly prepare the desired state~$\ket{\psi_1}$.
A similar result can be easily obtained if we relax both the assumption and the outcome to be approximate.

\begin{lemma}[Amplitude amplification]
\label{lem:AA}
    Suppose we are given a $\Gamma$-symmetric circuit~$\Acal$ such that
    
    $$\Acal\ket{0}^n = \sqrt{p} \ket{\psi_1} + \sqrt{1-p} \ket{\psi_0},$$
    where $\ket{\psi_0}$ and $\ket{\psi_1}$ are orthogonal states and $p\in (0, 1)$ is known.
    Suppose we are also given a $\Gamma\times \{id_1\}$-symmetric circuit $\Dcal$ which distinguishes $\ket{\psi_1}$ from $\ket{\psi_0}$ in the following sense:
    $$\Dcal \ket{\psi_0}\ket{0} = \ket{\psi_0}\ket{0}, \quad \Dcal \ket{\psi_1}\ket{0} = \ket{\psi_1}\ket{1}.$$
    Then we can construct a $\Gamma$-symmetric circuit $\Ccal$ which maps $\ket{0}^n$ to $\ket{\psi_1}$, using $O\big((s(\Acal)+s(\Dcal))/\sqrt{p}\big)$ gates and $a(\Acal) + a(\Dcal)+3$ workspace qubits.
\end{lemma}

\begin{proof}
    We employ the usual double-reflection strategy, alternately reflecting through the ``good state'' we wish to obtain and a ``starting state'' we can prepare which has nonnegligible overlap with the good state.
    To arrive exactly at the state $\ket{\psi_1}$ at the end of the algorithm, we need to perform a well-chosen rotation $R_y(\beta)$ on an ancilla qubit and incorporate this ancilla into the good state to be reflected over.
    Implementing the reflections will in turn require two extra ancilla qubits, one starting at $\ket{0}$ for using in the ``distinguisher'' circuit $\Dcal$ and one starting at $\ket{-}:=\frac{\ket{0}-\ket{1}}{\sqrt{2}}$ for performing the phase kickback trick.

    Denoting $K := \big\lceil \frac{\pi}{4\sqrt{p}} \big\rceil$ and defining $\beta\in (0, \pi)$ to be the unique solution to the equation
    $$\sin\frac{\beta}{2} = \frac{1}{\sqrt{p}} \sin\Big(\frac{\pi}{4K+2}\Big),$$
    our quantum circuit~$\Ccal$ is given in Figure~\ref{fig:AA} below.
    
    \begin{figure}[ht]
        \centering
        \includegraphics[width=15cm]{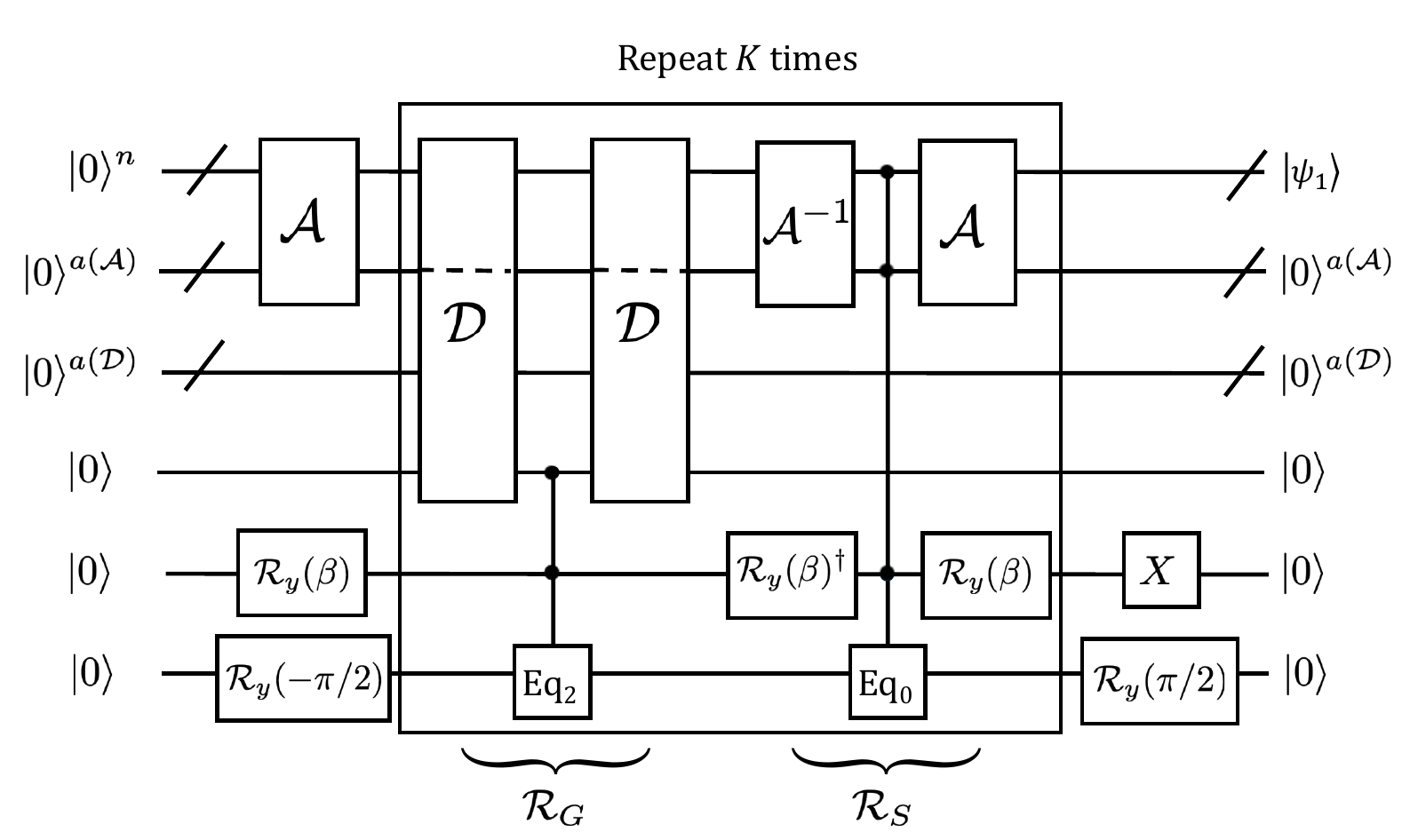}
        \caption{The circuit for symmetric amplitude amplification, with the reflections~$\Rcal_G$ and~$\Rcal_S$ highlighted.}
        \label{fig:AA}
    \end{figure}

    By Lemmas~\ref{lem:concat} and~$\ref{lem:inversion}$, this circuit is $\Gamma$-symmetric.
    We next outline the argument for its correctness.
    For notational convenience, we ignore the $a(\Acal)+a(\Dcal)$ workspace qubits which stay at state~$\ket{0}$ throughout our analysis.
    
    Denote by~$\ket{S}$ the state when we enter the ``repeat'' box, which we call the \emph{starting state};
    note that
    $$\ket{S} = \big(\sqrt{p} \ket{\psi_1} + \sqrt{1-p} \ket{\psi_0}\big) \ket{0} \big(R_y(\beta)\ket{0}\big) \ket{-}.$$
    The \emph{good state} we wish to obtain at the end of the ``repeat'' box is
    $$\ket{G} = \ket{\psi_1} \ket{0}\ket{1}\ket{-},$$
    so that we get the state $\ket{\psi_1}\ket{0}^3$ after applying the final layer of the circuit.
    Denote the two-dimensional space spanned by the starting state and the good state by~$V$.
    Let $\theta = \pi/(4K+2)$, so that $\braket{G | S} = \sin\theta$, and define the state $\ket{B} \in V$ by
    $$\ket{S} = \sin\theta \ket{G} + \cos\theta \ket{B}.$$

    Let~$\Rcal_G$ and~$\Rcal_S$ be the operators defined in Figure~\ref{fig:AA}.
    Restricted to the subspace~$V$, these operators act as reflections:
    \begin{align*}
        \Rcal_G \ket{G} &= -\ket{G}, &\Rcal_G \ket{B} = \ket{B}, \\
        \Rcal_S \ket{S} &= -\ket{S}, &\Rcal_S \ket{S^{\perp}} = \ket{S^{\perp}}
    \end{align*}
    where $\ket{S^{\perp}}\in V$ is orthogonal to~$\ket{S}$.
    The standard double-reflection argument (as presented in \cite{BrassardHMT2000}) shows that
    $$(\Rcal_S \Rcal_G)^j \ket{S} = \sin\big((2j+1)\theta\big) \ket{G} + \cos\big((2j+1)\theta\big) \ket{B}$$
    for all integers $j\geq 0$.
    Thus $(\Rcal_S \Rcal_G)^K \ket{S} = \ket{G}$, as wished.
\end{proof}

A useful variant of the amplitude amplification technique -- called \emph{oblivious amplitude amplification} -- was proposed in~\cite{BerryCCKS}.
It can be thought of as a generalisation of the previous technique where one replaces the starting state~$\ket{0}^n$ by an arbitrary (and possibly unknown) $n$-qubit state~$\ket{\psi}$.
Next, we show that this variant can also be performed in a symmetric way.

\begin{lemma}[Oblivious amplitude amplification]
\label{lem:obl_AA}
    Let $U$ be a $\Gamma$-symmetric $n$-qubit unitary.
    Suppose we are given a $\Gamma\times \{id_a\}$-symmetric $(n+a)$-qubit circuit $\Acal$ with action
    $$\Acal\ket{\psi} \ket{0}^a = \sqrt{p} (U\ket{\psi}) \ket{0}^a + \sqrt{1-p} \ket{\Phi_{\psi}} \quad \text{for all $\ket{\psi}$,}$$
    where $p\in (0, 1)$ is known and independent of $\ket{\psi}$, while $\ket{\Phi_{\psi}}$ is some normalised state that depends on $\ket{\psi}$ and has no support on basis states ending with $0^a$.
    Then we can construct a $\Gamma$-symmetric circuit which implements the unitary $U$ using $O\big(s(\Acal)/\sqrt{p}\big)$ gates and $a+2$ workspace qubits.
\end{lemma}

\begin{proof}
    The proof is similar to the one of the last lemma.
    
    Denoting $K := \big\lceil \frac{\pi}{4\sqrt{p}} \big\rceil$ and defining $\gamma\in (0, \pi)$ to be the unique solution to the equation
    $$\cos\frac{\gamma}{2} = \frac{1}{\sqrt{p}} \sin\Big(\frac{\pi}{4K+2}\Big),$$
    our quantum circuit for oblivious amplitude amplification is shown below in Figure~\ref{fig:oblivious_AA}.

    \begin{figure}[ht]
        \centering
        \includegraphics[width=14cm]{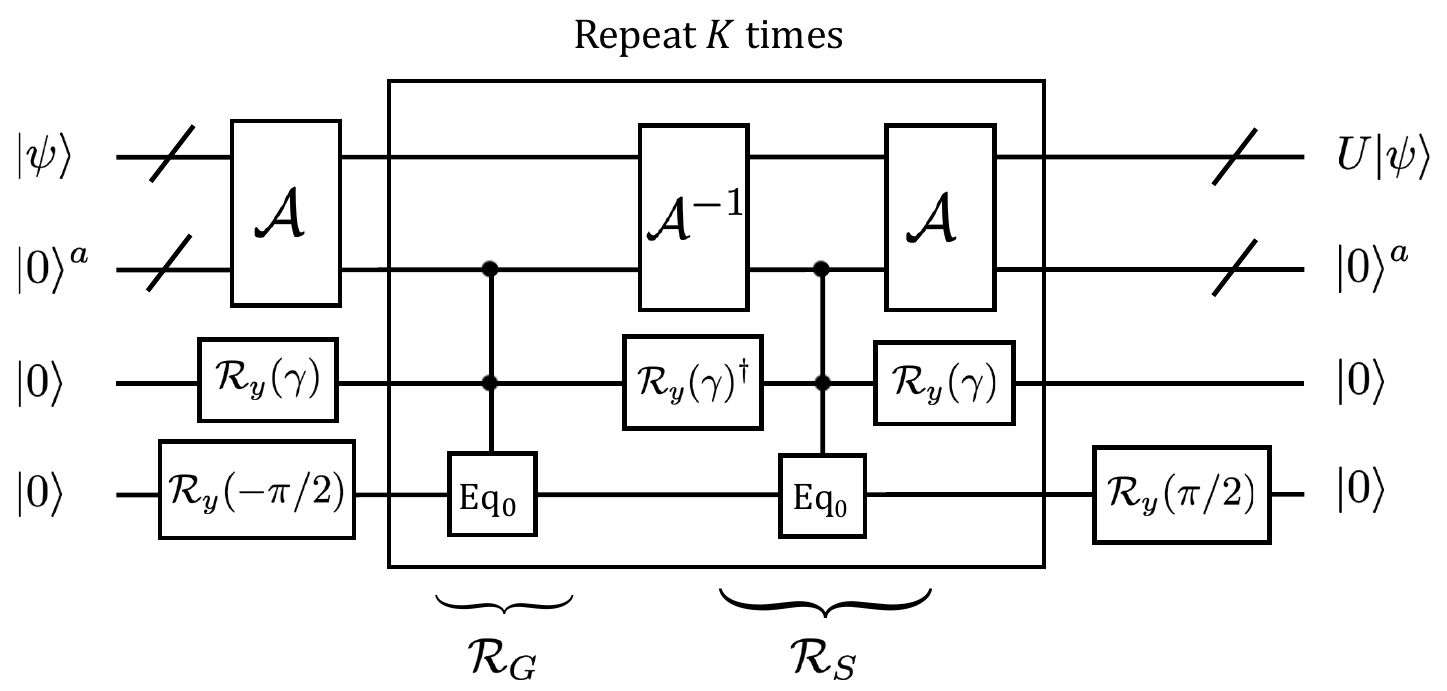}
        \caption{The circuit for oblivious amplitude amplification, with the reflections~$\Rcal_G$ and~$\Rcal_S$ highlighted.}
        \label{fig:oblivious_AA}
    \end{figure}

    Given an $n$-qubit quantum state~$\ket{\psi}$ on which our circuit acts, define the \emph{good state}
    $$\ket{G_{\psi}} = \big(U\ket{\psi}\big) \ket{0}^{a+1} \ket{-}$$
    and the \emph{starting state}
    $$\ket{S_{\psi}} = \big(\Acal \ket{\psi}\ket{0}^a\big) \big(R_y(\gamma) \ket{0}\big) \ket{-}.$$
    Denote by~$V_{\psi}$ the two-dimensional subspace spanned by~$\ket{G_{\psi}}$ and~$\ket{S_{\psi}}$.
    Let $\theta = \pi/(4K+2)$, so that $\braket{G_{\psi} | S_{\psi}} = \sin\theta$, and define the state $\ket{B_{\psi}} \in V_{\psi}$ by
    $$\ket{S_{\psi}} = \sin\theta \ket{G_{\psi}} + \cos\theta \ket{B_{\psi}}.$$

    Consider the operators~$\Rcal_G$ and~$\Rcal_S$ given in Figure~\ref{fig:oblivious_AA}.
    One can easily check that, restricted to the subspace~$V_{\psi}$, $\Rcal_G$ acts as a reflection through the line perpendicular to~$\ket{G_{\psi}}$:
    $$\Rcal_G \ket{G_{\psi}} = -\ket{G_{\psi}}, \quad \Rcal_G \ket{B_{\psi}} = \ket{B_{\psi}}.$$
    Moreover, within~$V_{\psi}$, $\Rcal_S$ acts as a reflection through the line perpendicular to~$\ket{S_{\psi}}$:
    $$\Rcal_S \ket{S_{\psi}} = -\ket{S_{\psi}}, \quad \Rcal_S \ket{S_{\psi}^{\perp}} = \ket{S_{\psi}^{\perp}}$$
    for $\ket{S_{\psi}^{\perp}} \in V_{\psi}$ orthogonal to~$\ket{S_{\psi}}$
    (see \cite[Lemma~3.7]{BerryCCKS}).

    It then follows from the standard double-reflection argument that the state obtained at the end of the ``repeat'' box is
    $$(\Rcal_S \Rcal_G)^K \ket{S_{\psi}} = \sin\big((2K+1)\theta\big) \ket{G_{\psi}} + \cos\big((2K+1)\theta\big) \ket{B_{\psi}} = \ket{G_{\psi}}.$$
    Applying a rotation $R_y(\pi/2)$ to the last qubit we obtain $(U\ket{\psi}) \ket{0}^{a+2}$, as desired.
\end{proof}

\subsection{Phase estimation}

Phase estimation is an efficient algorithmic procedure for estimating an eigenvalue of a known unitary~$U$ when given access to its corresponding eigenvector~$\ket{\psi}$.
It was first proposed by Kitaev \cite{Kitaev1995}, and later considered in  broader context in~\cite{CleveEMM1998}.

Suppose~$U$ is a unitary that we can implement via some quantum circuit~$\Acal$, and let~$\ket{\psi}$ be an eigenvector of~$U$.
Note that we can write $U\ket{\psi} = e^{2\pi i \theta} \ket{\psi}$ for some (unknown) phase $\theta\in [0, 1)$.
Quantum phase estimation allows us to estimate the value of this phase~$\theta$ up to precision~$\varepsilon$ using $O(1/\varepsilon)$ applications of control-$\Acal$, $O(1/\varepsilon^2)$ extra elementary gates and a single copy of the eigenstate~$\ket{\psi}$.

More precisely, denoting $k = \lceil\log(1/\varepsilon)\rceil$, what this procedure allows us to do is implement a circuit that maps $\ket{0}^k\ket{\psi}$ to a superposition
\begin{equation} \label{eq:phase_est}
    \bigg(\sum_{t=0}^{2^k-1} \alpha_t \ket{t}\bigg) \ket{\psi}, \quad \text{where} \quad |\alpha_t| = \bigg|\frac{\sin((2^k\theta-t)\pi)}{2^k \sin((\theta-t/2^k)\pi)}\bigg|
\end{equation}
(see \cite[Section~5]{CleveEMM1998}).
This gives an approximation of~$\theta$ because the amplitudes~$|\alpha_t|$ are highly concentrated around $t = 2^k \theta$.
In particular, measuring the first register of~\eqref{eq:phase_est} yields an integer~$t$ which satisfies\footnote{Here we use the distance $\|\cdot\|_{\R/\Z}$ in the torus since the phases $\theta=0$ and $\theta=1$ are  equivalent.}
\begin{equation} \label{eq:good_est}
    \|\theta - t/2^k\|_{\R/\Z} < \varepsilon
\end{equation}
with probability at least~$4/5$.

The following result shows that we can (efficiently) perform the phase estimation procedure in a way that also respects the symmetries of the circuit~$\Acal$.

\begin{lemma}[Phase estimation]
    Let~$U$ be a $\Gamma$-symmetric $n$-qubit unitary and let~$\ket{\psi}$ be an eigenstate of~$U$ such that $U\ket{\psi} = e^{2\pi i \theta} \ket{\psi}$ for some unknown phase $\theta\in [0, 1)$.
    Let $0<\varepsilon<1$ and denote $k = \lceil \log(1/\varepsilon)\rceil$.
    Given a $\Gamma$-symmetric circuit~$\Acal$ which implements~$U$, we can construct an $(\{id_k\}\times \Gamma)$-symmetric circuit that maps $\ket{0}^k\ket{\psi} \mapsto \big(\sum_{t=0}^{2^k-1} \alpha_t \ket{t}\big) \ket{\psi}$ as in equation~\eqref{eq:phase_est},
    using $O\big(s(\Acal)/\varepsilon + 1/\varepsilon^2\big)$ gates and $O(a(\Acal)+h(\Acal))$ workspace qubits.
\end{lemma}

\begin{proof}
    The standard phase estimation algorithm is given as follows:
    \begin{itemize}
        \item[$(i)$] Apply $H^{\otimes k} \otimes I_{2^n}$ to $\ket{0}^k \ket{\psi}$;
        \item[$(ii)$] Apply a select operation $\ket{j}\ket{\psi} \mapsto \ket{j} U^j\ket{\psi}$ for $0\leq j< 2^k$;
        \item[$(iii)$] Apply QFT$_{2^k}^{\dagger} \otimes I_{2^n}$, where QFT$_{2^k}^{\dagger}$ is the inverse quantum Fourier transform on $\Z/2^k\Z$.
    \end{itemize}
    The computations done in \cite[Section~5]{CleveEMM1998} show that this algorithm maps $\ket{0}^n \ket{\psi}$ to the desired state given in~\eqref{eq:phase_est}.

    Steps~$(i)$ and~$(iii)$ of the algorithm above are already given by $(\{id_k\}\times \Gamma)$-symmetric circuits (for any $\Gamma \leq \Scal_n$).
    Since~$\Acal$ implements~$U$ and is $\Gamma$-symmetric, we can use the ``select'' circuit from Lemma~\ref{lem:select_op} with $\Ccal_j = \Acal^j$ to implement step~$(ii)$ symmetrically as well;
    note that we can reuse the same workspace qubits for all such circuits~$\Ccal_j$, as they have the same symmetries.
    The result follows from concatenating the circuits for each of these steps.
\end{proof}

Note that the probability of obtaining a good estimate for the desired phase~$\theta$ (in the sense of equation~\eqref{eq:good_est}) can be made to approach~$1$ using standard methods, which are also symmetry-preserving.

\subsection{Linear combination of unitaries}

An important technique in Hamiltonian simulation is known as \emph{linear combination of unitaries}, or LCU for short.
Intuitively, it gives a recipe for implementing any given linear combination $\sum_{i=0}^{k-1} \alpha_i U_i$ of unitaries $U_0, \dots, U_{k-1}$ when we know how to implement each of these unitaries individually.

This technique can also be seen as a more sophisticated ``building block'' for creating more complex quantum circuits from simpler ones.
This will be illustrated in Theorem~\ref{thm:sym_state} and Theorem~\ref{thm:qu_part_sym} in later sections.

As the final item of Theorem~\ref{thm:subroutines} stated in the Introduction, we now prove that LCU can be performed in such a way that respects the symmetries of the unitaries.

\begin{lemma}[LCU] \label{lem:LCU}
    Let $U_0, \dots, U_{k-1}$ be $\Gamma$-symmetric $n$-qubit unitaries, let $\alpha_0, \dots, \alpha_{k-1}$ be complex numbers and suppose $L := \big\|\sum_{i=0}^{k-1} \alpha_i U_i \ket{0}^n\big\|\neq 0$.
    Given $\Gamma$-symmetric circuits $\Ccal_0, \dots, \Ccal_{k-1}$ such that each $\Ccal_i$ implements $\ket{0}^n \mapsto U_i \ket{0}^n$,

    we can construct a $\Gamma$-symmetric circuit $\Ccal_{\textsc{lcu}}$ which implements
    $$\ket{0}^n \mapsto \frac{1}{L} \bigg(\sum_{i=0}^{k-1} \alpha_i U_i \bigg) \ket{0}^n$$
    using $O\big(\|\alpha\|_1 \big(k\log k + \sum_{i=0}^{k-1} s(\Ccal_i)\big)/L\big)$ gates and $O\big(n + \sum_{i=0}^{k-1}a(\Ccal_i)\big)$ workspace qubits.
\end{lemma}

\begin{proof}
    Denote $m:= \lceil\log k\rceil$.
    By Lemma \ref{lem:select_op}, we can construct an $\{id_m\}\times \Gamma$-symmetric circuit $\Ccal_{sel}$ which implements the select operation
    \begin{equation*}
        \ket{j} \ket{\psi} \mapsto
        \frac{\alpha_j}{|\alpha_j|} \ket{j} U_j\ket{\psi} \quad \text{for $0\leq j < k$}
    \end{equation*}
    using $O\big(n + \sum_{j=0}^{k-1}a(\Ccal_j)\big)$ ancilla qubits and $O\big(k\log k + \sum_{j=0}^{k-1} s(\Ccal_j)\big)$ gates.
    We can also construct an ancilla-free $m$-qubit circuit $W$ of size $O(k)$ such that
    $$W \ket{0}^m = \frac{1}{\sqrt{\|\alpha\|_1}} \sum_{j=0}^{k-1} \sqrt{|\alpha_j|} \ket{j};$$
    see \cite{MVBS2005} for how this can be achieved.

    Denote $a':= a(\Ccal_{sel})$ and $n' := n+a'$ for convenience, and consider the circuit
    $$\Acal = \big(W^\dagger \otimes I_{2^{n'}}\big) \circ \Ccal_{sel} \circ \big(W \otimes I_{2^{n'}}\big).$$
    This circuit is $\{id_m\}\times \Gamma$-symmetric (by Lemma~\ref{lem:concat}), it uses $O\big(n + \sum_{j=0}^{k-1}a(\Ccal_j)\big)$ workspace qubits and has size $O\big(k\log k + \sum_{j=0}^{k-1} s(\Ccal_j)\big)$.
    A simple computation shows that
    $$(\bra{0}^m \otimes I_{2^{n'}}) \Acal \ket{0}^m \ket{0}^{n'} = \frac{1}{\|\alpha\|_1} \bigg(\sum_{j=0}^{k-1} \alpha_j U_j \ket{0}^n\bigg) \ket{0}^{a'}.$$
    Denoting $\ket{\psi_1} = \frac{1}{L} \ket{0}^m \big(\sum_{j=0}^{k-1} \alpha_j U_j \ket{0}^{n}\big) \ket{0}^{a'}$, we conclude that
    $$\Acal \ket{0}^m \ket{0}^{n'} = \frac{L}{\|\alpha\|_1} \ket{\psi_1} + \sqrt{1 - \frac{L^2}{\|\alpha\|_1^2}} \ket{\psi_0}$$
    for some $(m+n')$-qubit state $\ket{\psi_0}$ which has no support on basis states starting with $0^m$.
    We can then use amplitude amplification (Lemma~\ref{lem:AA}) on circuit $\Acal$ (with the circuit $\Dcal$ being a single gate $\Eq_0$ applied to the first $m$ qubits) in order to conclude the proof.
\end{proof}

\begin{remark}
    The bound obtained on the number of workspace qubits can be significantly improved in some cases, for instance when the workspace qubits of the circuits $\Ccal_j$ satisfy the same symmetries, or when there are few qubits which serve as the head of some threshold gate.
\end{remark}

In the previous lemma we have used in an essential way the fact that we started with a known state (in our case~$\ket{0}^n$), and the obtained circuit does not necessarily implement the linear combination $\frac{1}{L} \sum_i \alpha_i U_i$ for any other $n$-qubit states.
This is unavoidable in general, as the linear combination $\frac{1}{L} \sum_i \alpha_i U_i$ might not even be a unitary matrix.

In the special case where this linear combination of unitaries is also unitary, we can overcome this obstacle as shown in the following lemma:

\begin{lemma}[Oblivious LCU] \label{lem:obl_LCU}
    Let $U_0, \dots, U_{k-1}$ be $n$-qubit unitaries, let $\alpha_0, \dots, \alpha_{k-1}$ be complex numbers and suppose that $\sum_{i=0}^{k-1} \alpha_i U_i$ is also unitary.
    Given $\Gamma$-symmetric circuits $\Ccal_0, \dots, \Ccal_{k-1}$ which implement $U_0, \dots, U_{k-1}$, we can construct a $\Gamma$-symmetric circuit $\Ccal_{\textsc{lcu}}$ which implements $\sum_{i=0}^{k-1} \alpha_i U_i$ using
    $O\big(\|\alpha\|_1 \big(k\log k + \sum_{i=0}^{k-1} s(\Ccal_i)\big)\big)$
    gates and $O\big(n + \sum_{j=0}^{k-1}a(\Ccal_j)\big)$ workspace qubits.
\end{lemma}

\begin{proof}
    Denote $m:= \lceil\log k\rceil$.
    We first construct an ancilla-free $m$-qubit circuit $W$ of size $O(k)$ such that
    $$W \ket{0}^m = \frac{1}{\sqrt{\|\alpha\|_1}} \sum_{j=0}^{k-1} \sqrt{|\alpha_j|} \ket{j}.$$
    By Lemma \ref{lem:select_op}, we can construct an $\{id_m\}\times \Gamma$-symmetric circuit $\Ccal_{sel}$ which implements
    \begin{equation*}
        \ket{j} \ket{\psi} \mapsto
        \frac{\alpha_j}{|\alpha_j|} \ket{j} U_j\ket{\psi} \quad \text{for $0\leq j < k$}
    \end{equation*}
    using $O\big(n + \sum_{j=0}^{k-1}a(\Ccal_j)\big)$ workspace qubits and $O\big(k\log k + \sum_{j=0}^{k-1} s(\Ccal_j)\big)$ gates.

    Denote $a' := a(\Ccal_{sel})$, $n' := n+a'$ and consider the circuit
    $$\Acal = \big(W^\dagger \otimes I_{2^{n'}}\big) \circ \Ccal_{sel} \circ \big(W \otimes I_{2^{n}}\big).$$
    A simple computation shows that
    $$\Acal \ket{0}^m \ket{\psi} \ket{0}^{a'} = \frac{1}{\|\alpha\|_1} \ket{0}^m \bigg(\sum_{j=0}^{k-1} \alpha_j U_j \ket{\psi} \bigg) \ket{0}^{a'} + \sqrt{1- \frac{1}{\|\alpha\|_1^2}} \ket{\Phi_\psi},$$
    where $\big(\bra{0}^m \otimes I_{2^{n'}}\big) \ket{\Phi_\psi} = 0$.
    Using oblivious amplitude amplification (Lemma~\ref{lem:obl_AA}), we can then construct a $\Gamma$-symmetric circuit which implements $\sum_{j=0}^{k-1} \alpha_j U_j$ using $O\big(\|\alpha\|_1 s(\Ccal_{sel})\big)$ gates and $m+a'+2$ workspace qubits.
    The result follows.
\end{proof}

\section{Symmetric state preparation}
\label{sec:state_prep}

One of the key subroutines in quantum computing is the task of \emph{state preparation}.
There are several approaches for realising this task: adiabatic algorithms, variational methods (e.g., VQE), and designing quantum circuits that directly prepare the state.
We consider the latter option: one is given a classical description of a quantum state $\ket{\phi}$ and is asked to generate a quantum circuit that prepares $\ket{\phi}$ from a fixed state $\ket{0}^n$.

This task is notoriously difficult in general: preparing an arbitrary $n$-qubit quantum state using only single- and two-qubit gate requires a circuit of depth $\Theta(n)$ and an exponential amount of ancillary qubits (in $n$)~\cite{zhang2022quantum}.

Symmetries can dramatically reduce the number of gates required to prepare that state. For example, the Greenberger-Horne-Zeilinger (GHZ) states can be prepared very efficiently: a GHZ state can be prepared using a linear number of CNOT gates~\cite{bennett1996mixed} in the number of qubits, as opposed to an exponential number of gates that might be required for an arbitrary entangled state. Similarly, preparing the well-known AKLT state~\cite{smith2024constant} requires only a constant-depth quantum circuit. Improvements that emerge from using symmetries for this task have been studied in the context of variational algorithms~\cite{gard2020efficient}.

If the desired class of quantum states satisfies a non-trivial group of symmetries, one might wonder whether its state-preparation circuit can be made to satisfy those same symmetries.

\subsection{Permutation symmetry}

A rich source of applications of symmetries to quantum information theory comes from considering the \emph{symmetric subspace} $\vee^n \C^2$.
This is the vector space formed by all $n$-qubit states which are invariant under permutations of the qubits, and has been proven useful in state estimation, optimal cloning and the quantum de Finetti theorem \cite{Harrow2013}.

Clearly, any quantum state that can be prepared by an $\Scal_n$-symmetric circuit when starting from $\ket{0}^n$ must belong to~$\vee^n \C^2$.
An immediate question is whether the opposite holds:
can every symmetric state $\ket{\phi}\in \vee^n \C^2$ be prepared by some permutation-symmetric quantum circuit?
Moreover, which of these states can be \emph{efficiently} prepared by a permutation-symmetric quantum circuit?

Remarkably, it turns out that the obvious necessary condition of symmetry is sufficient even for ensuring the efficiency of a symmetric state-preparation circuit.
In other words, in the context of $\Scal_n$-symmetric quantum circuits, there is an equivalence between states that are \emph{possible} to prepare and states that are \emph{efficient} to prepare.
As mentioned above, this is far from true in the general (unrestricted) circuit setting.

Towards proving this result, we first consider the special case of \emph{Dicke states}.
Given integers $n\geq k\geq 0$, the Dicke state $\ket{D^n_k}$ is defined by
$$\Ket{D^n_k} := \binom{n}{k}^{-1/2} \sum_{x\in \bset{n}:\: |x|=k} \ket{x}.$$

\begin{lemma}[Dicke state preparation] \label{lem:Dicke}
    Given integers $n\geq k\geq 0$, there exists an $\Scal_n$-symmetric circuit that maps $\ket{0}^n$ to $\ket{D^n_k}$ using 3 workspace qubits and $O(n k^{1/4})$ gates.
\end{lemma}

\begin{proof}
We can assume that $1\leq k \leq n/2$, as the case $k=0$ is trivial and the case where $k > n/2$ is analogous.
Let $\theta = 2\arcsin{\sqrt{k/n}}$, so that
$$R_y(\theta) =
\begin{pmatrix}
    \sqrt{1-k/n} & -\sqrt{k/n} \\
    \sqrt{k/n} & \sqrt{1-k/n}
\end{pmatrix}.$$
Applying $R_y(\theta)$ to each of~$n$ qubits initialised $\ket{0}$, we obtain the state
$$R_y(\theta)^{\otimes n} \ket{0}^n = \sum_{x\in \bset{n}} \big(\sqrt{k/n}\big)^{|x|} \big(\sqrt{1-k/n}\big)^{n-|x|} \ket{x}.$$
Note that $\bra{D^n_k} R_y(\theta)^{\otimes n} \ket{0}^n = \sqrt{p}$, where
$$p = \binom{n}{k} \Big(\frac{k}{n}\Big)^{k} \Big(1 - \frac{k}{n}\Big)^{n-k};$$
it follows that we can write
$$R_y(\theta)^{\otimes n} \ket{0}^n = \sqrt{p} \ket{D^n_k} + \sqrt{1-p} \ket{\psi_0},$$
where $\ket{\psi_0}$ has zero amplitude on computational basis states of weight~$k$.

Applying symmetric amplitude amplification (Lemma~\ref{lem:AA}) to the circuit $\Acal = R_y(\theta)^{\otimes n}$ and with an $\Eq_k$ gate performing the part of the distinguisher circuit~$\Dcal$, we can construct an $\Scal_n$-symmetric circuit which maps~$\ket{0}^n$ to~$\ket{D^n_k}$ using~$3$ workspace qubits and $O(n/\sqrt{p})$ gates.
This circuit is shown in Figure~\ref{fig:Dicke}.

\begin{figure}[ht]
    \centering
    \includegraphics[width=15cm]{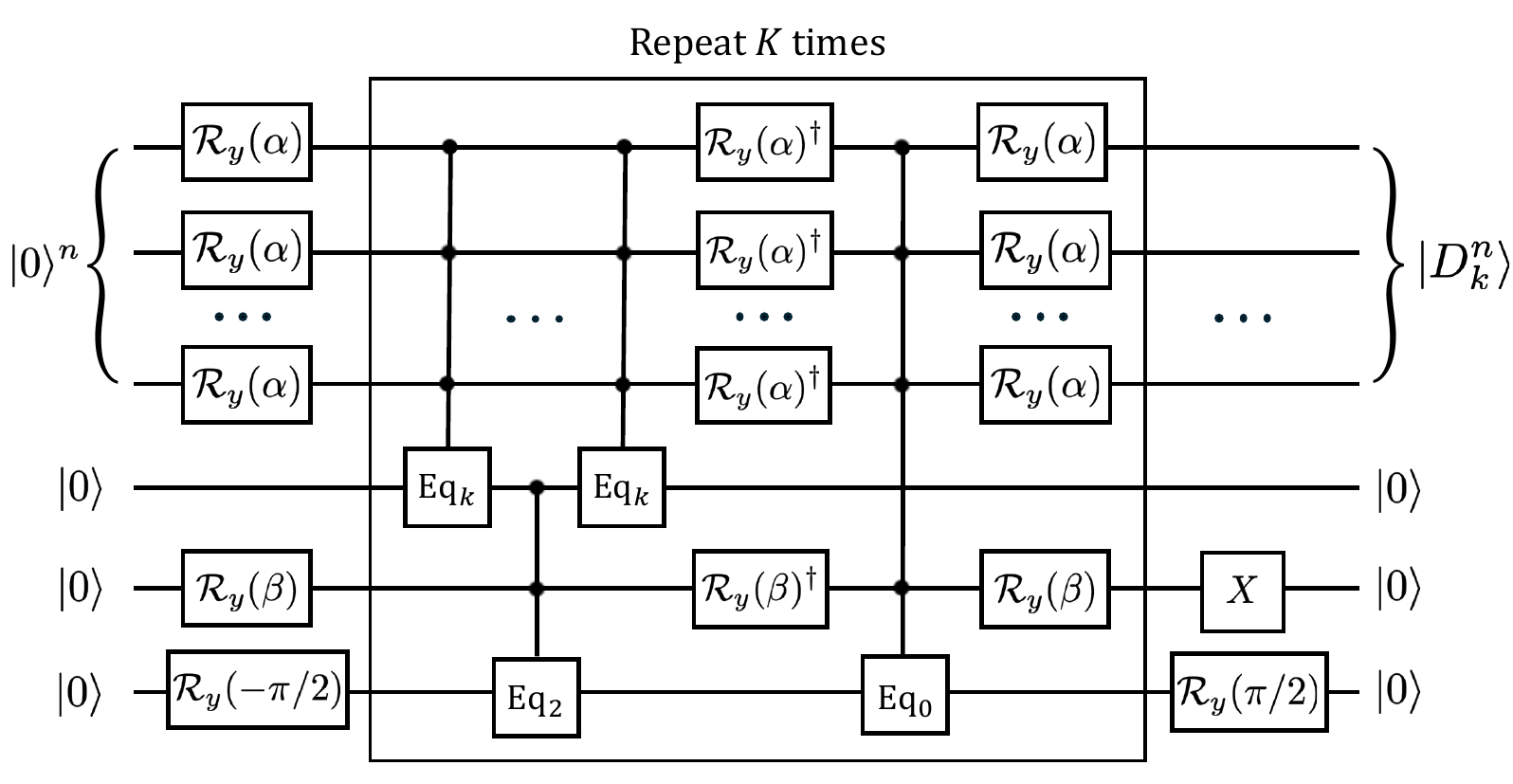}
    \caption{An $\Scal_n$-symmetric circuit for preparing the Dicke state~$\ket{D^n_k}$.}
    \label{fig:Dicke}
\end{figure}

It remains to estimate the value of~$\sqrt{p}$.
This can be done using the binomial inequalities
$$\sqrt{\frac{n}{8 k(n-k)}} 2^{n H(k/n)} \leq \binom{n}{k} \leq \sqrt{\frac{n}{2\pi k(n-k)}} 2^{n H(k/n)}$$
given in \cite[Chapter~10, Lemma~7]{MacWilliamsS1977}, where $H(\cdot)$ denotes the binary entropy function.
By definition, we can write $p = \binom{n}{k} 2^{-n H(k/n)}$, from which we conclude that $p = \Theta(k^{-1/2})$.
The number of gates in the circuit is then $\Theta(nk^{1/4})$, as wished.
\end{proof}

\begin{theorem}[Symmetric state preparation] \label{thm:sym_state}
    Given the classical description of a symmetric $n$-qubit state~$\ket{\varphi}$, we can construct an $\Scal_n$-symmetric quantum circuit~$\Ccal$ which prepares $\ket{\varphi}$ starting from $\ket{0}^n$ using $O(n^{2.75})$ gates and~$6$ workspace qubits.
\end{theorem}

\begin{proof}
Since the set of Dicke states $\big\{\ket{D^n_k}:\: 0\leq k \leq n\big\}$ forms an orthonormal basis of the symmetric subspace, we can decompose any symmetric $n$-qubit state $\ket{\varphi}$ as
$$\ket{\varphi} = \sum_{k=0}^n \alpha_{k} \Ket{D^n_k},$$
where $\alpha_0, \dots, \alpha_n$ are complex numbers with $|\alpha_0|^2 + \dots + |\alpha_n|^2 = 1$.
By the Cauchy-Schwarz inequality, we have
$$\sum_{k=0}^n |\alpha_k| \leq \sqrt{n+1}.$$

For each $0 \leq k \leq n$, denote by $\Ccal_k$ the symmetric circuit that maps $\ket{0}^n$ to $\ket{D^n_k}$ constructed in Lemma~\ref{lem:Dicke}.
These circuits use 3 ancilla qubits and $O(n^{5/4})$ gates.
Using linear combination of unitaries (Lemma~\ref{lem:LCU}), we can construct an $\Scal_n$-symmetric circuit $\Ccal$ such that
$$\Ccal \ket{0}^n = \bigg(\sum_{k=0}^n \alpha_k \Ccal_k\bigg) \ket{0}^n = \ket{\varphi},$$
and which uses $O(n^{11/4})$ gates and $O(n)$ workspace qubits.
\end{proof}

\subsection{Towards general symmetries}

One can also consider quantum states which are symmetric with respect to other symmetry groups.
An interesting class of examples are the \emph{polynomial phase states} studied in \cite{ArunachalamBSY2023}, which contain as special cases graph states, hypergraph states and states produced by measurement-based quantum computing.

\begin{definition}[Polynomial phase state]
    Given a polynomial $P: \F_2^n \to \F_2$, the phase state associated to~$P$ is
    $$\ket{\varphi_P} := 2^{-n/2} \sum_{x\in \bset{n}} (-1)^{P(x)} \ket{x}.$$
\end{definition}

Denote by $\Pol_{\leq d}(\F_2^n)$ the set of all polynomials of degree at most~$d$ on~$n$ variables over~$\F_2$, and let $$[n]_{\leq d} := \{J\subseteq [n]:\, |J|\leq d\}.$$
There is a natural bijection between $\Pol_{\leq d}(\F_2^n)$ and the set $\bset{[n]_{\leq d}}$ given by the \emph{list of coefficients}:
for every $P\in \Pol_{\leq d}(\F_2^n)$ there is a unique list $(a_J:\, J\in [n]_{\leq d})$ such that
$$P(x_1, \dots, x_n) = \sum_{J\subseteq [n]:\, |J|\leq d} a_J \prod_{j\in J} x_j \mod{2}.$$
Denote by~$\ket{P} = \bigotimes_{J\in [n]_{\leq d}} \ket{a_J}$ the computational basis state encoding this list of coefficients.

The group of symmetries of the set $[n]_{\leq d}$ corresponds to the pointwise action of the permutation group~$\Scal_n$ on subsets $J\subseteq [n]$ of size at most~$d$;
denote this group by~$\Scal^n_{\leq d}$.
We next show how to construct a linear-size $\Scal^n_{\leq d}$-symmetric circuit which, on input $\ket{P} \in (\C^2)^{[n]_{\leq d}}$, prepares the corresponding phase state~$\ket{\varphi_P}$.

\begin{theorem}[Phase state preparation]
    Given integers $n, d \geq 1$, we can construct an $\Scal^n_{\leq d}$-symmetric circuit~$\Ccal$ of size $O(n^d)$ such that
    $$\Ccal \ket{P} \ket{0}^n = \ket{P} \ket{\varphi_P} \quad \text{for all $P\in \Pol_{\leq d}(\F_2^n)$.}$$
\end{theorem}

\begin{proof}
    We start by applying a layer of Hadamard gates to all qubits, obtaining the state
    $$\bigg(\bigotimes_{J\in [n]_{\leq d}} \frac{\ket{0} + (-1)^{a_J}\ket{1}}{\sqrt{2}}\bigg) \otimes \bigg(\frac{1}{2^{n/2}} \sum_{x\in \F_2^n} \ket{x}\bigg).$$
    We then apply an~$X$ gate to the wire indexed by~$\emptyset$ and, for each set $J\subseteq[n]$ of size $1\leq |J|\leq d$, we apply the threshold gate\footnote{The gate $\thr^{J, J}_{\geq |J|}$ amounts to a multi-controlled NOT gate where the control qubits are indexed by the elements in~$J$ and the target qubit is indexed by the set~$J$.} $\thr^{J, J}_{\geq |J|}$.
    Note that all of these gates pairwise commute, and the resulting state after the second layer is
    $$\bigg(\bigotimes_{J\in [n]_{\leq d}} \frac{\ket{0} + (-1)^{a_J}\ket{1}}{\sqrt{2}}\bigg) \otimes \bigg(\frac{1}{2^{n/2}} \sum_{x\in \F_2^n} \prod_{J\in [n]_{\leq d}} (-1)^{a_J \prod_{j\in J} x_j} \ket{x}\bigg).$$
    Finally, we end with a layer of Hadamard gates applied to all input wires (i.e., those indexed by $[n]_{\leq d}$);
    this produces the state
    $$\bigg(\bigotimes_{J\in [n]_{\leq d}} \ket{a_J}\bigg) \otimes \bigg(\frac{1}{2^{n/2}} \sum_{x\in \F_2^n} \prod_{J\in [n]_{\leq d}} (-1)^{a_J \prod_{j\in J} x_j} \ket{x}\bigg) = \ket{P} \ket{\varphi_P}.$$
    The circuit thus constructed is easily checked to be $\Gamma$-symmetric.
\end{proof}

Finally, another family of symmetric states originates from the solutions to certain symmetric problems.

Indeed, suppose we have a decision problem on $n$-bit strings which is invariant under the action of some symmetry group $\Gamma \leq \Scal_n$.
In that case, the uniform superposition over all accepting bit strings (written in the computational basis) will be a $\Gamma$-symmetric quantum state.
The following result says that we can prepare this quantum state symmetrically whenever we can solve the associated decision problem symmetrically.

\begin{theorem}[Superposition over accepting states]
    Let $f: \bset{n} \to \bset{}$ be a $\Gamma$-symmetric Boolean function, and denote $p:= |f^{-1}(\{1\})|/2^n$.
    Given a $\Gamma$-symmetric circuit~$\Ccal$ which implements~$f$, we can construct a $\Gamma$-symmetric circuit that prepares the state
    $$\ket{\psi_f} := \frac{1}{\sqrt{p2^n}} \sum_{x: f(x)=1} \ket{x}$$
    using $O\big((n+s(\Ccal))/\sqrt{p}\big)$ gates and $a(\Ccal)+3$ workspace qubits.
\end{theorem}

\begin{proof}
    This immediately follows by applying symmetric amplitude amplification (Lemma~\ref{lem:AA}) to circuits $\Acal = H^{\otimes n}$ and $\Dcal = \Ccal$, with $\ket{\psi_1} = \ket{\psi_f}$.
\end{proof}

\section{Symmetric quantum advantage} \label{sec:advantage}

In this section, we show how the decision problem XOR-SAT can be decided by polynomial-size symmetric quantum circuits.
This implies an exponential quantum advantage in the symmetric setting, since it is known that XOR-SAT on~$n$ variables and $m\geq n$ constraints requires $2^{\Omega(n)}$-sized symmetric threshold circuits to solve.
We will first formally define the decision problem XOR-SAT, then explain how this classical lower bound can be obtained, and finally provide our efficient symmetric quantum algorithm.

Consider a set of~$m$ linear equations over~$n$ Boolean variables $x_1, x_2, \dots, x_n \in \{0,1\}$, where each equation is of the form
\[
x_{i_1} \oplus x_{i_2} \oplus \cdots \oplus x_{i_k} = b,
\]
where $\oplus$ denotes the exclusive-or (addition modulo $2$), $i_1, \dots, i_k \in [n]$, and $b \in \{0,1\}$.
Equivalently, one may represent an instance of a problem as a system of linear equations over the finite field $\mathrm{GF}(2)$:
\[
A \mathbf{x} = \mathbf{b} \pmod{2},
\]
where $A \in \{0,1\}^{m \times n}$, $\mathbf{x} \in \{0,1\}^n$, and $\mathbf{b} \in \{0,1\}^m$.
Given the pair $(A,\mathbf{b})$ as input, the XOR-SAT problem asks:
does there exist an assignment $\mathbf{x} \in \{0,1\}^n$ that satisfies all equations simultaneously?
This is a special case of the Boolean satisfiability problem where each clause is an XOR (parity) constraint.
Note that this problem is invariant under permutation of constraints and permutation of variables, that is, it is $\Scal_m \times \Scal_n$-symmetric.

It was first established by Atserias, Bulatov and Dawar \cite{AtseriasBD2009} that XOR-SAT is not expressible in fixed-point logic with counting (FPC), which implies super-polynomial lower bounds on the size of families of symmetric circuits solving this problem \cite{AndersonD2017}.
This result was later strengthened by Atserias and Dawar \cite{AtseriasD19}:
even the easier problem of separating satisfiable instances of XOR-SAT (on~$n$ variables and $m = \Theta(n)$ constraints) from those that are not $(1/2+\varepsilon)$-satisfiable, any FPC-formula requires $\Omega(n)$ variables.
This in turn implies a lower bound of $2^{\Omega(n)}$ for the size of $\Scal_m \times \Scal_n$-symmetric threshold circuits that solve this (easier) approximate version of XOR-SAT on~$m = \Theta(n)$ constraints and~$n$ variables.

Here we show that symmetric quantum circuits can solve the XOR-SAT efficiently:

\begin{theorem}
    There is a $\Scal_m \times \Scal_n$-symmetric quantum circuit of size $\Tilde{O}(m^2 n^2)$ that solves XOR-SAT on~$m$ constraints and~$n$ variables with probability~1.
\end{theorem}

Our construction uses the fact that solvability of an XOR-SAT instance $(A, b)$ is invariant with respect to doubling constraints (i.e. repeating some row of $(A, b)$) and to doubling variables (i.e. repeating some column of~$A$).
Indeed, a doubled constraint has the same truth value as the original one, while doubling some variable~$y$ to $(y', y'')$ is equivalent to considering the same constraints acting on a new variable $z := y' \oplus y''$.

The main idea needed for the circuit is to symmetrically create a uniform superposition of (slightly larger) modified instances $(A', b')$, each of which is solvable if and only if the original instance $(A, b)$ is solvable.
This superposition state is created from the original instance by entangling the columns of matrix~$A$ to $O(n\log n)$ Dicke states $\ket{D_1^n}$, each of which ``chooses'' a column to be copied into~$A'$, and then doing a similar procedure with the rows of $(A,b)$.
We can apply amplitude amplification to make sure all rows and columns of $(A,b)$ are copied into $(A',b')$ at least once.
By the fact given in the last paragraph, this implies that each modified instance $(A', b')$ with nonzero amplitude in the resulting state will have the same XOR-SAT truth value as the original instance.
We then solve the XOR-SAT problem coherently (using Gaussian elimination, say) in this superposition state and output the answer.

The crucial insight is that the final superposition state is invariant with respect to permuting rows and columns of the instance $(A,b)$, and thus performing Gaussian elimination on it is an $\Scal_m \times \Scal_n$-symmetric operation.
Since each step used in the creation of the superposition state can be implemented by an efficient symmetric quantum circuit (by Lemmas~\ref{lem:AA} and~\ref{lem:Dicke}), the outlined algorithm gives rise to an efficient symmetric quantum circuit.
The details are given below.

\begin{proof}
Let~$X$ be a set of~$m$ elements labelling the constraints and~$Y$ be a set of~$n$ elements labelling the variables, so that the problem is $\Scal_X \times \Scal_Y$-symmetric.
Denote the XOR-SAT instance by $(A, b)$, where $A\in \bset{X\times Y}$ and $b\in \bset{X}$, and let $\Tilde{m} = \lfloor10 m \log m\rfloor$, $\Tilde{n} = \lfloor10 n \log n\rfloor$.

We will implement the strategy outlined above, with the quantum circuit being constructed step by step.
For $i\in [\Tilde{m}]$, $j\in [\Tilde{n}]$, denote
$$R(i) = \big\{\text{row}(i; x):\: x\in X\big\} ,\quad C(j) = \big\{\text{col}(j; y):\: y\in Y\big\},$$
and let $R = \bigcup_{i\in [\Tilde{m}]} R(i)$ and $C = \bigcup_{j\in [\Tilde{n}]} C(j)$.
The circuit will have workspace qubits labelled by the elements of $R\cup C$, as well as other workspace qubits to be specified later.
By Lemma~\ref{lem:Dicke}, for each $i\in [\Tilde{m}]$ we can construct an $O(m)$-size $\Scal_m$-symmetric quantum circuit on qubits $R(i)$ that maps $\ket{0}^m$ to $\ket{D^m_1}$.
(Each such circuit introduces three extra ancillas starting and ending at~$\ket{0}$, which we will ignore as they have no significant impact.)
Likewise, for each $j\in [\Tilde{n}]$ we construct an $O(n)$-size $\Scal_n$-symmetric quantum circuit on qubits $C(j)$ that maps $\ket{0}^n$ to $\ket{D^n_1}$.
The resulting state on qubits $R\cup C$ after this is
$$\bigotimes_{i\in [\tilde{m}]} \bigg(\frac{1}{\sqrt{m}}\sum_{x'\in X} \bigotimes_{x\in X} \ket{\one\{x=x'\}}_{\text{row}(i;x)}\bigg) \bigotimes_{j\in [\tilde{n}]} \bigg(\frac{1}{\sqrt{n}}\sum_{y'\in Y} \bigotimes_{y\in Y} \ket{\one\{y=y'\}}_{\text{col}(j;y)}\bigg).$$
Intuitively, having value~$\ket{1}$ in qubit $\text{row}(i;x)$ indicates that row $x$ of $(A, b)$ has been chosen to be the $i$-th row of the modified instance $(A', b')$, and value~$\ket{1}$ in qubit $\text{col}(j;y)$ means that column $y$ of~$A$ has been chosen to be the $j$-th column of~$A'$.

Next we check whether all elements $x\in X$, $y\in Y$ have been chosen at least once, which can be done as follows.
Introduce new qubits with labels $\text{use}(x)$ for all~$x$, $\text{use}(y)$ for all~$y$ and $\text{use}(all)$, each initialised~$\ket{0}$.
Apply gates $\thr^{S_x, \text{use}(x)}_{\geq 1}$ and $\thr^{S_y, \text{use}(y)}_{\geq 1}$, where
$$S_x = \big\{\text{row}(i;x):\: i\in [\tilde{m}]\big\}, \quad S_y = \big\{\text{col}(j;y):\: j\in [\tilde{n}]\big\},$$
for all $x\in X$ and $y\in Y$.
Note that all of these gates pairwise commute.
Finally, apply the gate $\thr^{\text{Use, use}(all)}_{\geq m+n}$ where
$$\text{Use} = \big\{\text{use}(x):\: x\in X\big\} \cup \big\{\text{use}(y):\: y\in Y\big\},$$
and uncompute the qubits in $\text{Use}$ by applying gates $\thr^{S_x, \text{use}(x)}_{\geq 1}$, $x\in X$ and $\thr^{S_y, \text{use}(y)}_{\geq 1}$, $y\in Y$ a second time.
Note that the circuit obtained up to now is symmetric with respect to the natural action of $\Scal_X\times \Scal_Y$, where $(\sigma, \tau) \in \Scal_X\times \Scal_Y$ maps the qubits with labels $\text{row}(i;x)$, $\text{col}(j;y)$, $\text{use}(x)$ and $\text{use}(y)$ to those with labels $\text{row}(i; \sigma(x))$, $\text{col}(j; \tau(y))$, $\text{use}(\sigma(x))$ and $\text{use}(\tau(y))$, respectively.

If we were to measure qubit $\text{use}(all)$ and obtain~$\ket{1}$, the state~$\ket{\phi_{RC}}$ on qubits $R\cup C$ would collapse to the uniform superposition over all computational basis states of the form
\begin{equation} \label{eq:phiRC}
    \bigg(\bigotimes_{i\in [\tilde{m}]} \bigotimes_{x\in X} \ket{\one\{x=x_i\}}_{\text{row}(i;x)}\bigg) \otimes \bigg(\bigotimes_{j\in [\tilde{n}]} \bigotimes_{y\in Y} \ket{\one\{y=y_j\}}_{\text{col}(j;y)}\bigg)
\end{equation}
where $(x_1, \dots, x_{\tilde{m}}) \in X^{\tilde{m}}$, $(y_1, \dots, y_{\tilde{n}}) \in Y^{\tilde{n}}$, every element of~$X$ appears at least once as some~$x_i$ and every element of~$Y$ appears at least once as some~$y_j$.
Moreover, by our choices of $\tilde{m}$, $\tilde{n}$, and the coupon collector theorem, the amplitude of~$\ket{1}$ on qubit $\text{use}(all)$ is larger than~$1/4$ (and we can compute it explicitly).
It follows from symmetric amplitude amplification (Lemma~\ref{lem:AA}) that we can construct an $\Scal_X\times \Scal_Y$-symmetric quantum circuit of size $\tilde{O}(m^2+n^2)$ that prepares the state $\ket{\phi_{RC}}$ defined above on qubits $R\cup C$.

Now we use the obtained state $\ket{\phi_{RC}}$ and the instance $(A,b)\in \bset{X\times Y}\times \bset{X}$ to obtain a symmetric superposition over modified instances $(A',b') \in \bset{\tilde{m}\times \tilde{n}}\times \bset{\tilde{m}}$ which all have the same XOR-SAT value as $(A,b)$.
Denote by $(x,y)$ and~$x$ the labels of the input qubits where the values of $A(x,y)$ and $b(x)$ are stored, respectively.
We introduce qubits labelled $(i,j)$ and~$i$ for $i\in [\tilde{m}]$, $j\in [\tilde{n}]$, which will store the values of $A'(i,j)$ and $b'(i)$.
We implement the gates $\thr^{S(x,y,i,j), (i,j)}_{\geq 3}$ and $\thr^{S(x,i), i}_{\geq 2}$ for all $x\in X$, $y\in Y$, $i\in [\tilde{m}]$ and $j\in [\tilde{n}]$, where
$$S(x,y,i,j) = \big\{(x,y),\, \text{row}(i; x),\, \text{col}(j; y)\big\}, \quad S(x,i) = \big\{x,\, \text{row}(i; x)\big\}.$$
Note that all of these gates pairwise commute and, in the component of $\ket{\phi_{RC}}$ given by equation~\eqref{eq:phiRC}, the states on qubits $(i,j)$ and~$i$ after these operations are
$$\ket{A'(i,j)} :=\ket{A(x_i, y_j)} \quad \text{and} \quad \ket{b'(i)} := \ket{b(x_i)},$$
respectively.
Since every label $x\in X$, $y\in Y$ appears at least once as some $x_i$ or $y_j$, and since XOR-SAT is invariant with respect to doubling constraints and variables, it follows that each instance $(A',b')$ with nonzero amplitude in the resulting state is satisfiable iff $(A,b)$ is satisfiable.

Finally, we make two remarks on the symmetries of the resulting quantum circuit.
First, it is symmetric with respect to the natural action of $\Scal_X \times \Scal_Y$;
this follows from our construction and the symmetric concatenation lemma (Lemma~\ref{lem:concat}).
Second, the action of $\Scal_X \times \Scal_Y$ on the qubits $(i,j)$ and~$i$ is trivial for all $i\in [\tilde{m}]$, $j\in [\tilde{n}]$:
since the set of all gates acting on any of these qubits is pairwise commuting and closed under permutations $(\sigma, \tau)\in \Scal_X \times \Scal_Y$, this group induces no action on these qubits.
As such, we can implement an $O(\tilde{n}^2 \tilde{m})$-size circuit that solves XOR-SAT on the instance $(A',b')$ encoded in the qubits $(i,j)$ and~$i$ (for $i\in [\tilde{m}]$, $j\in [\tilde{n}]$), and the final circuit thus obtained will also be $\Scal_m \times \Scal_n$-symmetric.
The main contribution to the size of the final circuit are the gates $\thr^{S(x,y,i,j), (i,j)}_{\geq 3}$, of which there are $mn \tilde{m}\tilde{n} = \tilde{O}(m^2n^2)$ many.
The result follows.
\end{proof}

\section{Examples of symmetric quantum algorithms}
\label{sec:examples}

We now show how to express some symmetric quantum algorithms using the formalism of symmetric quantum circuits.

\subsection{Group-equivariant quantum neural networks}

In the field of quantum machine learning, group-equivariant quantum neural networks (QNNs) have emerged as appropriate architectures for learning highly symmetric properties.
When incorporating the underlying invariances in the data as geometric priors into their learning architectures, one significantly restricts the space of functions explored by
the QNN model without weakening its predictive power.
Such models can lead to better performance, both in training and
generalization, than those without inductive biases \cite{ragone2022representation}.

The symmetry constraints imposed on group-equivariant QNNs are strongly connected to our notion of symmetry for quantum circuits.
Indeed, such QNNs can be regarded as particular instantiations of symmetric quantum circuits, albeit with different gate sets.

Here we illustrate how this difference in gate sets -- and in perspectives -- is of little regard.
We show that the permutation-equivariant QNNs proposed in \cite{SLNSC2024} can be naturally translated into $\Scal_n$-symmetric circuits with only a constant-factor increase on the number of gates (though at the expense of adding $\Theta(n^2)$ workspace qubits).
The procedure outlined here is not specific to the permutation group, and should work for more general group-equivariant QNNs as well.

The permutation-equivariant QNNs of \cite{SLNSC2024} consist of series of `layers' of parametrised unitaries, each of the form $e^{-i\theta \Hcal}$ for some $\Hcal\in \{\Hcal_X, \Hcal_Y, \Hcal_{ZZ}\}$, where
$$\Hcal_X = \frac{1}{n} \sum_{j=1}^n X_j, \quad \Hcal_Y = \frac{1}{n} \sum_{j=1}^n Y_j, \quad \Hcal_{ZZ} = \frac{2}{n(n-1)} \sum_{j<k} Z_j Z_k.$$
In addition, the observables considered in the context of QNNs are of the form
\begin{equation} \label{eq:obs}
    \frac{1}{n} \sum_{j=1}^n \chi_j, \quad \frac{2}{n(n-1)} \sum_{j<k} \chi_j \chi_k, \quad \prod_{j=1}^n \chi_j
\end{equation}
for some Pauli matrix~$\chi$.
We will show how to implement each layer of the QNN and each observable by an $\Scal_n$-symmetric quantum circuit.
Concatenating those circuits gives the desired symmetric implementation of the neural networks.

The layers corresponding to $\Hcal_X$ equate to the unitary
$$e^{-i\theta \Hcal_X} = \prod_{j=1}^n e^{-i(\theta/n) X_j} = \bigotimes_{j=1}^n e^{-i(\theta/n) X},$$
which is an $\Scal_n$-symmetric layer of~$n$ single-qubit gates.
Similarly for those layers corresponding to $\Hcal_Y$.
The layers corresponding to $\Hcal_{ZZ}$ consist of two-qubit gates $e^{-i\alpha Z_j Z_k}$ which do not directly belong to our gate set;
to implement them, we introduce $\binom{n}{2}$ workspace qubits.

For each unordered pair of distinct indices $j, k\in [n]$, we add a workspace qubit labelled $\{j,k\}$.
Note that
$$e^{-i\theta \Hcal_{ZZ}} = \prod_{j<k} e^{-i\frac{2\theta}{n(n-1)} Z_j Z_k},$$
and that all of the terms in this product commute.
Setting $\alpha = 2\theta/(n(n-1))$, the unitary $e^{-i\alpha Z_j Z_k}$ can be implemented as follows:
apply a single-qubit gate $e^{-\alpha}I$ to qubit $\{j,k\}$, followed by CNOTs from qubits~$j$ and~$k$ to qubit~$\{j,k\}$, followed by a phase gate $P_{2\alpha}$ on qubit $\{j,k\}$, followed by CNOTs from qubits~$j$ and~$k$ to qubit~$\{j,k\}$.
This is shown in Figure~\ref{fig:exp_ZZ}.

\begin{figure}[ht]
    \centering
    \includegraphics[width=12cm]{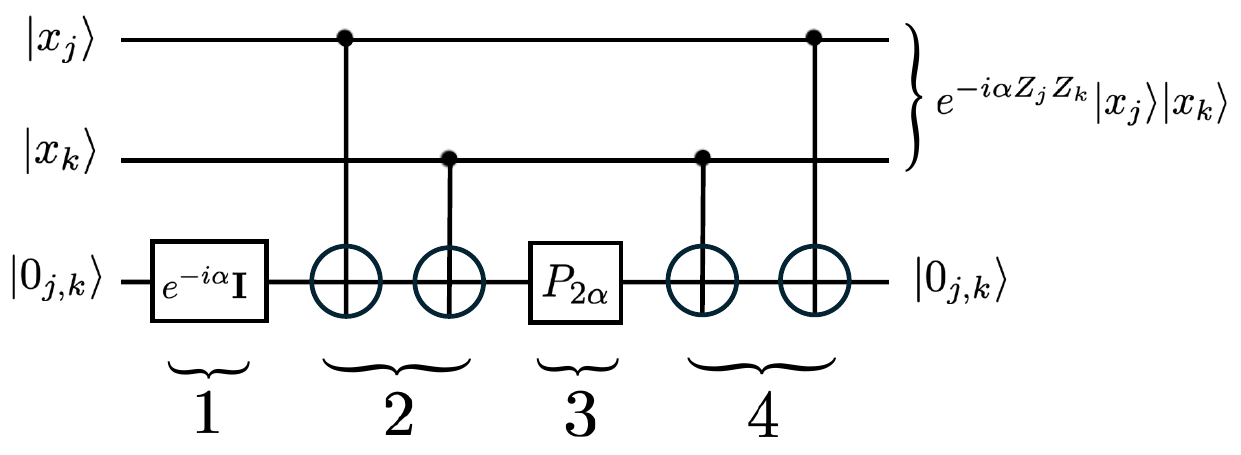}
    \caption{A circuit which implements $e^{-i\alpha Z_j Z_k}$.}
    \label{fig:exp_ZZ}
\end{figure}

Doing so for every pair of indices $j, k$ and grouping the resulting gates into four layers of analogous gates, we obtain an $\Scal_n$-symmetric circuit that implements $e^{-i\theta H_{ZZ}}$ using $O(n^2)$ gates and $O(n^2)$ workspace qubits (which are returned to $\ket{0}$ at the end).

For the observables, let us consider the case where the matrix~$\chi$ in Equation~\eqref{eq:obs} corresponds to~$Z$;
the other cases can be easily reduced to this one.
Note that, for any bit string $x\in \bset{n}$, we have
\begin{align*}
    \bigg(\frac{1}{n} \sum_{j=1}^n Z_j\bigg) \ket{x} &= \bigg(\frac{1}{n} \sum_{j=1}^n (-1)^{x_j}\bigg) \ket{x} = \bigg(1- \frac{2|x|}{n}\bigg) \ket{x}, \\
    \bigg(\prod_{j=1}^n Z_j\bigg) \ket{x} &= \bigg(\prod_{j=1}^n (-1)^{x_j}\bigg) \ket{x} = (-1)^{|x|} \ket{x}
\end{align*}
and
$$\bigg(\frac{2}{n(n-1)} \sum_{j<k} Z_j Z_k\bigg) \ket{x} = \bigg(\frac{2}{n(n-1)} \sum_{j<k} (-1)^{x_j+x_k}\bigg) \ket{x}.$$
The first two can be determined symmetrically using threshold-type gates.
For the third, we introduce $\binom{n}{2}$ extra workspace qubits to encode the values of $y_{i,j} := x_i\oplus x_j$.
We can then express the value of the observable on $\ket{x}$ as $1-4|y|/(n(n-1))$, which can also be symmetrically determined using threshold-type gates.

\subsection{Restriction to the symmetric subspace}

We have shown in Theorem~\ref{thm:sym_state} that every state in the symmetric subspace $\vee^n \C^2$ can be efficiently prepared by an $\Scal_n$-symmetric circuit.
We now show how these circuits can also efficiently implement any symmetric unitary when restricted to the symmetric subspace:

\begin{theorem}[Symmetric subspace unitaries] \label{thm:sym_subspace}
    Given an $\Scal_n$-symmetric unitary $U\in \U(2^n)$, we can implement the restricted action of~$U$ on the symmetric subspace as an $\Scal_n$-symmetric circuit with $O(n^{3.75})$ gates and $O(n)$ workspace qubits.
\end{theorem}

\begin{proof}
For $k\in \{0, \dots, n\}$, denote by $\Ccal_k$ the $\Scal_n$-symmetric circuit that implements $\ket{0}^n \mapsto \ket{D^n_k}$ constructed in Lemma~\ref{lem:Dicke}, and let $\ket{\phi_k} := U\ket{D^n_k}$.
Note that the following equations completely determine the action of~$U$ on the symmetric subspace:
$$U\ket{D^n_k} = \ket{\phi_k} \quad \text{for $0\leq k\leq n$.}$$
Denote by $\Acal_k$ the $\Scal_n$-symmetric circuit that implements $\ket{0}^n \mapsto \ket{\phi_k}$ constructed in Theorem~\ref{thm:sym_state}.

For each $k$, it follows that the circuit $\Acal_k \circ \Ccal_k^{\dagger}$ maps $\ket{D^n_k}$ to $\ket{\phi_k}$
(when those states are padded with the appropriate number of qubits $\ket{0}$).
We can then implement~$U$ on a symmetric state $\ket{\psi} = \sum_{k=0}^n \alpha_k \ket{D^n_k}$ by checking, for all $k\in \{0, \dots, n\}$, whether the first~$n$ qubits of $\Ccal_k^\dagger \ket{\psi}$ are zero and then implementing either $\Acal_k$ (if this is the case) or $\Ccal_k$ (if not).
The circuit in Figure~\ref{fig:sym_unitary}(a) does this symmetrically for a given~$k$ with the help of one extra ancilla qubit (using a new ancilla for each~$k$).

\begin{figure}[ht]
    \centering
    \includegraphics[width=15cm]{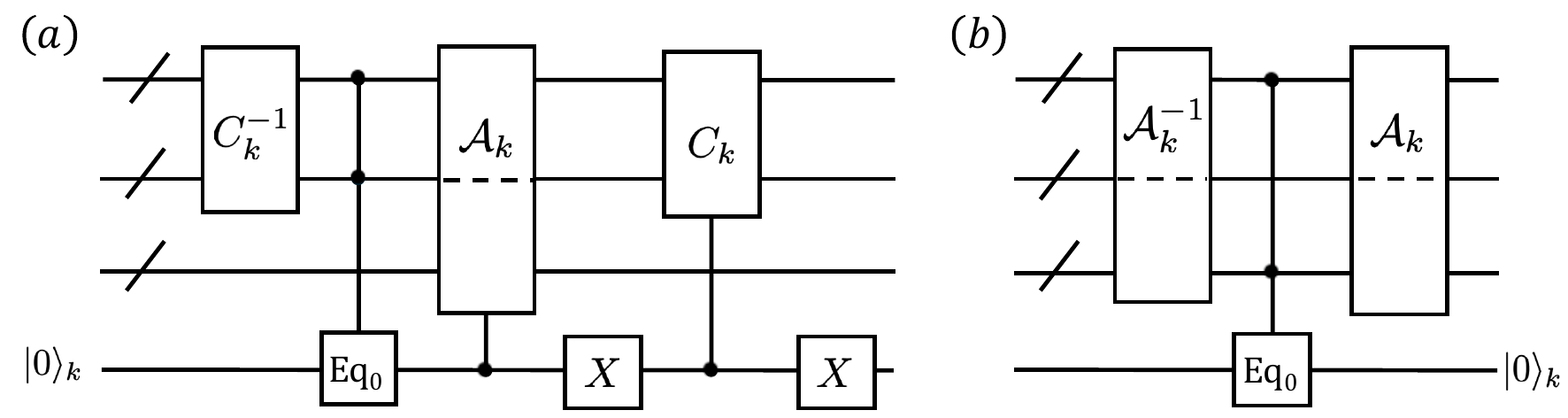}
    \caption{Parts of the circuit that implements the restriction of a unitary to the symmetric subspace.}
    \label{fig:sym_unitary}
\end{figure}

At the end of the computation we must return this ancilla to~$\ket{0}$, which can be done by applying $\Acal_k^\dagger$, then an $\Eq_0$ gate followed by the application of~$\Acal_k$ as shown in Figure~\ref{fig:sym_unitary}(b).
Implementing all circuits shown in Figure~\ref{fig:sym_unitary}(a) for $0\leq k\leq n$ (in some arbitrary order) and then all circuits in Figure~\ref{fig:sym_unitary}(b), we obtain the desired $\Scal_n$-symmetric circuit.
\end{proof}

\subsection{Partition-symmetric unitaries}

Even though we have seen that symmetric quantum circuits can perform many important quantum subroutines, an important complexity-theoretic question is whether they are \emph{universal} for symmetric operations.
This is not a trivial question:
Marvian showed that, in the context of \emph{continuous} symmetries and without allowing ancillas, generic symmetric unitaries cannot be implemented using local symmetric unitaries \cite{Marvian2022}.

Here we show that, in the context of any partition-symmetry group $\Gamma = \Scal_{n_1} \times \dots \times S_{n_t}$, $\Gamma$-symmetric quantum circuits are universal:
they can implement any $\Gamma$-symmetric unitary.

\begin{theorem} \label{thm:qu_part_sym}
    Any partition-symmetric unitary can be implemented by a partition-symmetric quantum circuit.
\end{theorem}

\begin{proof}
Let $\Gamma = \prod_{i=1}^t \Scal_{n_i}$ be the considered symmetry group, denote $n = n_1 +\dots+ n_t$ and let~$U$ be a $\Gamma$-symmetric $n$-qubit unitary.
By Corollary~\ref{cor:part_sym} we can write
$$U = \sum_{i_1, \dots, i_t = 1}^m \alpha_{i_1, \dots, i_t} V_{i_1}^{\otimes n_1} \otimes \dots \otimes V_{i_t}^{\otimes n_t},$$
where $m \leq \binom{n+3}{3}$, $V_i \in \U(2)$ for all $i\in [m]$ and
$$\sum_{i_1, \dots, i_t = 1}^m |\alpha_{i_1, \dots, i_t}| \leq C(n)$$
for some finite value $C(n) > 0$.
Using oblivious linear combination of unitaries (Lemma~\ref{lem:obl_LCU}), one can construct a $\Gamma$-symmetric circuit which implements~$U$ by combining the simple $\Gamma$-symmetric circuits $V_{i_1}^{\otimes n_1} \otimes \dots \otimes V_{i_t}^{\otimes n_t}$.
\end{proof}

\section{Acknowledgements}
We are grateful to Anuj Dawar for insightful discussions.
This work was supported by the Engineering and Physical Sciences Research Council on Robust and Reliable Quantum Computing (RoaRQ), Investigation 005 [grant reference EP/W032635/1]. TG acknowledges support by ERC Starting Grant 101163189 and UKRI Future Leaders Fellowship MR/X023583/1. SS acknowledges support from the Royal Society University Research Fellowship.

\printbibliography

\appendix

\section{Properties of symmetric reversible circuits}
\label{sec:appendix}

In this appendix we prove some important properties of symmetric reversible circuits.
In particular, we provide proofs to Proposition~\ref{prop:equiv} and Theorem~\ref{thm:part_sym} from Section~\ref{sec:reversible}.

We start with a useful observation about the action of reversible threshold gates within a single layer.
Let~$W$ be the set of wire labels in the circuit.
For each $i\in [k]$, let~$T_i$ be the threshold gate with support~$S_i$, head $h_i\notin S_i$ and parameter~$t_i$:
$$T_i:\: a\mapsto a\oplus (\one\{|a_{S_i}|\geq t_i\} e_{h_i}) \quad \text{for $a\in \bset{W}$,}$$
where $e_{h_i} \in \bset{W}$ denotes the bit string with a single~$1$ at index~$h_i$.
If $T_i$ and $T_j$ commute, by equating the action of $T_i\circ T_j$ and $T_j\circ T_i$ we conclude that
$$T_i\circ T_j(a) = a\oplus (\one\{|a_{S_i}|\geq t_i\} e_{h_i}) \oplus (\one\{|a_{S_j}|\geq t_j\} e_{h_j}).$$
More generally, if $T_1, \dots, T_k$ pairwise commute, it follows that
\begin{equation} \label{eq:commuting_thr}
    T_1\circ T_2\circ \cdots \circ T_k(a) = a\oplus \bigoplus_{i=1}^k \big(\one\{|a_{S_i}|\geq t_i\} e_{h_i}\big).
\end{equation}
Note that this same formula (when suitably interpreted in the ket notation) also holds for commuting threshold gates in \emph{quantum} circuits.

Recall that Proposition~\ref{prop:equiv} states that the Boolean (threshold) and reversible notions of symmetric circuits are equivalent, up to a linear increase in gate complexity.
This is proven by combining the two lemmas given below.

\begin{lemma}[Reversible implementation of Boolean circuits]
    Any $\Gamma$-symmetric threshold circuit with~$s$ gates can be converted into an equivalent $\Gamma$-symmetric reversible circuit which uses at most~$2s$ gates and~$s$ workspace bits.
\end{lemma}

\begin{proof}
    Let~$C$ be the $\Gamma$-symmetric threshold circuit in consideration, with vertex set~$V$, edge set~$E$ and gate labelling $\lambda: V \to X\cup G_{thr}$
    (where~$X$ denotes the structure on which the circuit acts).
    Let~$d$ be the depth of circuit~$C$, meaning the maximum length of a directed path in the underlying graph.
    For each $0\leq i\leq d$, denote by~$V_i \subseteq V$ the vertices at depth~$i$:
    those for which the maximum
    
    We will construct a $\Gamma$-symmetric layered reversible circuit~$R$ with wire labels $W := X\cup V$ and layers $L_0, L_1, \dots, L_d$ to be specified below.
    The wires labelled by the structure~$X$ encode the input, so that wire $x\in X$ starts with the same value as the input gate $v_x \in V$ with $\lambda(v_x) = x$ in~$C$, while the ``workspace'' wires labelled by~$V$ are initially~0.
    The reversible circuit~$R$ will be equivalent to~$C$ in the following way:
    for every gate $v\in V$, at the end of the computation, the bit value at wire~$v$ in~$R$ will be the same as the value at gate~$v$ in~$C$.
    
    The layer~$L_0$ is given by $L_0 = \big\{\Eq_1^{\{x\}, v_x}:\, x\in X \big\}$, where $v_x \in V_0$ is the vertex with $\lambda(v_x) = x$.
    For each $1\leq i\leq d$, we initialise~$L_i$ empty and then, for each $v\in V_i$:
    \begin{itemize}
        \item If~$v$ is a $\NOT$ gate in~$C$ with incoming edge $(u, v)$, add the gates $\NOT_v$ and $\Eq_1^{\{u\}, v}$ to~$L_i$.
        \item If~$v$ is a threshold gate $\thr^k_{\geq t}$ with incoming edges $(u_1, v), \dots, (u_k, v)$, add gate $\thr^{\{u_1, \dots, u_k\}, v}_{\geq t}$ to~$L_i$.
        Similarly for when~$v$ is an equality gate $\Eq^k_t$.
    \end{itemize}
    The reversible circuit thus obtained is easily checked to be $\Gamma$-symmetric and equivalent to~$C$.
\end{proof}

\begin{lemma}[Boolean implementation of reversible circuits]
    Any $\Gamma$-symmetric reversible circuit with~$s$ gates can be converted into an equivalent $\Gamma$-symmetric threshold circuit with at most~$4s$ (non-input) gates.
\end{lemma}

\begin{proof}
We start by giving a threshold implementation of the parity function, which will be needed afterwards.
Note that
$$x_1 \oplus \dots \oplus x_m = \bigvee_{k=0}^{\lfloor m/2 \rfloor} \big[|x| = 2k+1\big].$$
This formula shows that we can implement parity on~$m$ bits via a threshold circuit of depth~$2$ using one~$\Eq^m_{2k+1}$ gate for each $0\leq k \leq \lfloor m/2\rfloor$ and one $\thr^{\lfloor m/2\rfloor +1}_{\geq 1}$ gate.
Moreover, this circuit is easily seen to be invariant under permutations of the~$m$ input bits.

Now let~$R$ be a $\Gamma$-symmetric reversible circuit over some structure~$X$, let $W$ be its set of wire labels (with $X\subseteq W$) and let $L_1, \dots, L_d$ be its layers.
Up to increasing the number of layers by at most a factor of three, we may assume that each of them contains only gates of the same type:
either $\NOT$, threshold or equality.
Our goal is to construct a threshold circuit~$C$ which is $\Gamma$-symmetric and equivalent to~$R$;
we do so one layer at a time, as follows.

Let $V_0 = \{(w,0): w\in W\}$ and initialise the edge set~$E$ to be empty.
The input gates of~$C$ will be the~$(x, 0)$ with $x\in X$, which will be labelled by their corresponding element in~$X$.
For each $w\in W\setminus X$, add edges $\big\{\big((x,0), (w,0)\big): x\in X\big\}$ to~$E$ and label the vertex $(w,0)$ by~$\Eq^{|X|}_{|X|+1}$.
(This is done so that the bits labelled by $(w,0)$ have value~$0$ as in the reversible circuit.)

Next, for $i=1, 2, \dots, d$ in order, initialise~$V_i$ empty and define the function $\ell_i: W\to \N$ by
$$\ell_i(w) = \max \{j\leq i-1:\, (w,j)\in V_j\}.$$
Then $\ell_i(w)=j$ means the following:
the last time before layer~$L_i$ that wire~$w$ was the head of some gate was in layer~$L_j$.
(If $\ell_i(w)=0$, then~$w$ was never the head of a gate before layer~$L_i$.)
The value of bit $(w, \ell_i(w))$ will be (by construction) the same as the value at wire~$w$ in the circuit~$R$ immediately before layer~$L_i$.

If $L_i$ is a $\NOT$ layer, then for each gate $\NOT_h \in L_i$ we add a vertex $(h,i)$ to~$V_i$ with label~$\NOT$ and add an edge $\big((h, \ell_i(h)), (h,i)\big)$ to~$E$.
Note that $|V_i| = |L_i|$.

If~$L_i$ is a threshold layer, then for each gate $\thr^{S,h}_{\geq t}$ in~$L_i$:
\begin{itemize}
    \item Add vertex $(h, S, t, i)$ to~$V_i$ and label it by $\thr^{|S|}_{\geq t}$;
    \item Add edges $\big\{\big((w, \ell_i(w)), (h,S,t,i)\big): w\in  S\big\}$ to~$E$.
\end{itemize}
Next, for each $h\in W$ which is the head of some gate in~$L_i$, add a vertex $(h,i)$ to~$V_i$ and implement a parity gate (as explained above) from vertices
$$\{(h, \ell_i(h))\} \cup \big\{(h,S,t,i):\, \thr^{S,h}_{\geq t} \in L_i \big\}$$
to vertex $(h,i)$.
By equation~\eqref{eq:commuting_thr}, the value at vertex~$(h,i)$ in this circuit will be the same as the value at wire~$h$ immediately after layer~$L_i$.
The case where~$L_i$ is an equality layer is analogous.
Note that $|V_i| \leq 3 |L_i|$.

The circuit~$C$
obtained at the end of this process is equivalent to~$R$:
the final value at any wire~$w$ in circuit~$R$ is equal to the value at vertex $(w, \ell_{D+1}(w))$ in circuit~$C$
(assuming both circuits have the same input, and defining~$\ell_{D+1}$ in the obvious way).
The number of non-input gates in this circuit is
$$\sum_{i=0}^d |V_i| - |X| \leq |W\setminus X| + \sum_{i=1}^d 3 |L_i| \leq 4s.$$
Finally, circuit~$C$ is also $\Gamma$-symmetric:
any automorphism $\pi\in \Scal_W$ of~$R$ can be mapped to an automorphism~$\sigma$ of~$C$ which permutes the vertices according to how~$\pi$ permutes their labels
(so that $\sigma(w,i) = (\pi(w), i)$ and $\sigma(h, S, t, i) = (\pi(h), \pi(S), t, i)$, for instance).
\end{proof}

Next we consider Theorem~\ref{thm:part_sym}, which concerns efficient symmetric implementations of partition-symmetric functions.
Recall that a function $F: \prod_{i=1}^k \bset{n_i} \to \prod_{i=1}^k \bset{n_i}$ is partition-symmetric if its action commutes with the natural action of the group $\prod_{i=1}^k \Scal_{n_i}$.

\begin{theorem}[Theorem~\ref{thm:part_sym} restated]
    Let $k\geq 1$ and $n_1, \dots, n_k \geq 1$ be integers, and denote $n = n_1 +\dots + n_k$.
    Any partition-symmetric reversible function
    $F: \prod_{i=1}^k \bset{n_i} \to \prod_{i=1}^k \bset{n_i}$
    can be implemented by a partition-symmetric reversible circuit using \\
    $O\big(n\prod_{i=1}^k (n_i+1)\big)$ gates and $O(n)$ workspace bits, all of which are returned to zero at the end of the computation.
\end{theorem}

To prove this theorem, we first introduce some notation.
An element $a\in \prod_{i=1}^k \bset{n_i}$ can be decomposed into the direct product $a_1 \times a_2 \times \dots \times a_k$, where $a_i \in \bset{n_i}$ for $i\in [k]$.
We denote by~$a_{i,j}$ the $j$-th element of~$a_i$ in this representation (where $j\in [n_i]$).

\begin{lemma}
    If $F: \prod_{i=1}^k \bset{n_i} \to \prod_{i=1}^k \bset{n_i}$ is partition-symmetric and reversible, then there exist functions $f_i: \prod_{j=1}^k \{0, \dots, n_j\} \to \bset{}$, $1\leq i \leq k$, such that
    $$F(a)_{i,j} = a_{i,j} \oplus f_i\big(|a_1|, \dots, |a_k|\big) \quad \text{for all $i\in [k]$, $j\in [n_i]$.}$$
\end{lemma}

\begin{proof}
For given indices $i\in [k]$, $j\in [n_i]$, we consider a collection of functions $v_{\ell}: \bset{} \times\{0, \dots, n_{\ell}\} \to \bset{n_\ell}$, $\ell\in [k]$, where each~$v_{\ell}(b, w)$ is an $n_{\ell}$-bit string of weight~$w$ and the $j$-th bit of $v_i(b,w)$ is~$b$.
Any such collection of functions would do, but for concreteness we can define them by
\begin{align*}
    v_{\ell}(b, w) &= (1^w, 0^{n_\ell - w}) \quad \text{if $\ell\neq i$,} \\
    v_i(b, w) &=
    \begin{cases}
        1 &\text{at the first $w-b$ indices in $[n_i]\setminus\{j\}$,} \\
        0 &\text{at the last $n_i-w+b-1$ indices in $[n_i]\setminus\{j\}$,} \\
        b &\text{at index $j$.}
    \end{cases}
\end{align*}
(The function $v_i$ is only well-defined when $0\leq w-b \leq n-1$, but this is the only range where we use it.)

Now define the function $g_{i,j}: \bset{} \times \prod_{\ell=1}^k \{0, \dots, n_{\ell}\} \to \bset{}$ by
$$g_{i,j}(b, w_1, \dots, w_k) = F\big(v_1(b, w_1), \dots, v_k(b, w_k)\big)_{i,j}.$$
Note that, by symmetry of~$F$, we must have $g_{i,j} \equiv g_{i,j'}$ for all $j, j'\in [n_i]$.
Finally, define $f_i: \prod_{\ell=1}^k \{0, \dots, n_{\ell}\} \to \bset{}$ by
$$f_i(w_1, \dots, w_k)=
\begin{cases}
    g_{i,1}(0, w_1, \dots, w_k) &\text{if $w_i \leq n_i-1$,} \\
    1\oplus g_{i,1}(1, w_1, \dots, w_k) &\text{if $w_i = n_i$.}
\end{cases}$$

By reversibility of~$F$, we must have $g_{i,1}(0, w_1, \dots, w_k) \neq g_{i,1}(1, w_1, \dots, w_k)$ whenever $1\leq w_i \leq n_i -1$, and thus
$$g_{i,1}(1, w_1, \dots, w_k) = 1\oplus g_{i,1}(0, w_1, \dots, w_k).$$
We conclude from the symmetry of~$F$ that, for all $a\in \prod_{i=1}^k \bset{n_i}$, we have
\begin{align*}
    F(a_1, \dots, a_k)_{i,j} &= F\big( v_1(a_{i,j}, |a_1|),\, \dots,\, v_k(a_{i,j}, |a_k|) \big)_{i,j} \\
    &= g_{i,1}(a_{i,j},\, |a_1|,\, \dots,\, |a_k|) \\
    &= a_{i,j} \oplus f_i(|a_1|, \dots, |a_k|),
\end{align*}
as wished.
\end{proof}

With the help of this characterization, we can now prove Theorem~\ref{thm:part_sym}.

\begin{proof}[Proof of Theorem~\ref{thm:part_sym}]
For each $i\in[k]$, let $f_i: \prod_{\ell=1}^k \{0, \dots, n_{\ell}\} \to \bset{}$ be the function for which
$$F(a_1, \dots, a_k)_{i,j} = a_{i,j} \oplus f_i(|a_1|, \dots, |a_k|) \quad \text{for all $j\in [n_i]$.}$$
Let $S_i := f_i^{-1}(\{1\}) \subseteq \prod_{\ell=1}^k \{0, \dots, n_{\ell}\}$, and note that
\begin{equation} \label{eq:f_i}
    f_i(w_1, \dots, w_k) = \sum_{s\in S_i} \one\big\{w_\ell = s_\ell \text{ for all $\ell\in[k]$}\big\}.
\end{equation}

We can then construct a partition-symmetric reversible circuit for implementing~$F$ as follows.
The input wires are labelled $(i,j)$ for $i\in[k]$, $j\in[n_i]$, and the workspace wires are labelled $w(i,t)$ for $i\in[k]$, $t\in \{0, \dots, n_i\}$.
Given $s \in \prod_{\ell=1}^k \{0, \dots, n_{\ell}\}$, denote $W(s) = \{w(\ell, s_\ell): \ell\in [k]\}$.
The gates of the first layer are given by $\Eq^{\{i\}\times[n_i],\, w(i,t)}_t$ for $i\in [k]$ and $t\in \{0, \dots, n_i\}$, so that on input $a\in \prod_{i=1}^k \bset{n_i}$ each wire $w(i, t)$ will have bit value $\one\{|a_i|=t\}$.
The second layer has gates $\Eq_k^{W(s),\, (i,j)}$ for $i\in [k]$, $j\in [n_i]$ and $s\in S_i$.
Equation~\eqref{eq:f_i} implies that this circuit computes~$F$, and one can clean the workspace bits with the usual uncomputation technique.
\end{proof}

\end{document}